\documentclass{article}

\usepackage{bm}
\usepackage{fullpage}
\usepackage{lipsum}
\usepackage{amsfonts, mathrsfs}
\usepackage{graphicx}
\usepackage{epstopdf}
\usepackage{multirow}
\usepackage{algorithm}
\usepackage{algpseudocode}
\usepackage{color}
\usepackage{amsthm}
\usepackage{amsmath}

\usepackage{amsopn}

\newcommand{\gia}[1]{{\color{black}{#1}}}

\newcommand{\gdm}[1]{{\color{black} #1}}
\newcommand{\dk}[1]{{\color{black} #1}}
\newcommand{\gst}[1]{{\color{black} #1}}
\newcommand{\cR}{{\mathcal R}}
\newcommand{\cS}{{\mathcal S}}
\newcommand{\known}{{\tt SloFI}}
\newcommand{\unknown}{{\tt SPoRD}}
\newcommand{\unknownack}{{\tt SPoRDAck}}
\newcommand{\E}{{\rm I\kern-.3em E}}
\newcommand{\lN}{\ln N }
\newcommand{\remove}[1]{}

\begin{document}

\title{Deterministic non-adaptive contention resolution\\ on a shared channel
	\thanks{A preliminary version of this paper  appeared in the proceedings of the 39th IEEE International Conference on Distributed Computing Systems (ICDCS 2019).}
}

\author{%
	Gianluca De Marco\footnotemark[4]
	\and
	Dariusz R.~Kowalski\footnotemark[3]
	\and
	Grzegorz Stachowiak\footnotemark[2]
}

\date{}

\maketitle

\footnotetext[4]{
	Dipartimento di Informatica,
	University of Salerno,
	Italy.
	Email: \texttt{gidemarco@unisa.it}.
}

\footnotetext[2]{
	Institute of Computer Science, University of Wroc{\l}aw, Poland.
	Email: \texttt{gst@cs.uni.wroc.pl}.
}

\footnotetext[3]{School of Computer and Cyber Sciences, Augusta University, GA, USA.
	Email: \texttt{dkowalski@augusta.edu}}
	
\maketitle

\newtheorem{fact}{Fact}[section]

\newtheorem{theorem}{Theorem}[section]
\newtheorem{lemma}{Lemma}[section]
\newtheorem{corollary}{Corollary}[section]
\newtheorem{claim}{Claim}[section]
\newtheorem{proposition}{Proposition}[section]
\newtheorem{definition}{Definition}[section]
\newtheorem{example}{Example}[section]

\begin{abstract}
In a multiple access channel, autonomous stations are able 
\gia{to transmit and listen to a shared device}. A fundamental problem, called \textit{contention resolution}, is to allow any station to successfully deliver its message by resolving the conflicts that arise when several stations transmit simultaneously. 
Despite a long history on such a problem, most of the results deal with the static
setting when all stations start simultaneously, while many fundamental questions remain open in the realistic scenario when stations can join 
the channel at arbitrary times.  

In this paper, we explore the impact that three major channel features (asynchrony among stations, knowledge of the number of contenders and possibility of switching off
stations after a successful transmission) can have on the time complexity of non-adaptive 
deterministic algorithms. 
We establish upper and lower bounds allowing to understand which parameters permit time-efficient contention resolution and which do not. 

\end{abstract}

\providecommand{\keywords}[1]{\textbf{\textit{Key words---}} #1}
\keywords{multiple-access channel, contention resolution, deterministic algorithm, lower bound.}

\section{Introduction}
A shared channel, also called a \textit{multiple access channel}, is one of the fundamental communication models. 
In this model, $N$ autonomous computing entities (called stations) are attached
to a shared medium. 
Only a subset of $k < N$ stations 
(called \textit{active} stations) have packets to be transmitted. 
These $k$ stations can join the channel independently at different times
as new packet transmission requests arrive (\textit{asynchronous activation times})
and \gia{can transmit messages} in synchronous rounds.
In a nutshell, the \textit{contention resolution} problem is to 
\gdm{allow} 
every active station to \gdm{eventually} transmit its packet successfully.

One of the main difficulties is how to efficiently resolve collisions occurring 
when more than one entity attempts to access the common resource at the same time.
The problem is particularly challenging for 
\textit{deterministic} distributed algorithms,
for which not much work has been done so far, as the lack of randomness
makes breaking symmetry among the contending stations more difficult and time
consuming. 

This work focuses on answering the following crucial questions regarding
deterministic contention resolution:
\begin{itemize}
	\item 
	Which impact does the asynchrony among activation times have
	on {time complexity}?
	\item
	How important is the knowledge/estimate of the number of contenders?
	\item
	Could non-adaptive protocols or codes (see Section~\ref{s:conclusions} for further discussion)
	be asymptotically as efficient as adaptive protocols?
\end{itemize}

These questions have been already answered, to some extent, in the case of
randomized distributed algorithms~\cite{Bend-20,DS-17}; this works
addresses the case of {\em deterministic} distributed solutions.
Before presenting the results of this work, we need a more formal 
definition of both the model and the problem under consideration.

\subsection{The Model}

The formal model that is taken as the basis for theoretical studies 
on contention resolution,
is defined as follows, \textit{cf.} the surveys by Gallager~\cite{Gal} and Chlebus~\cite{Chl},
as well as some recent work for randomized solutions~\cite{Bend-20,DS-17}.

\paragraph{\textbf{Stations.}} 
A set of $N$ {\em stations}, also called {\em nodes},
are connected to the same  
transmission medium (called a {\em shared channel}).
Stations have distinct IDs in the range $[N]=\{0,1,\ldots,N-1\}$.  
Only up to $k$ stations, out of these $N$, might become active
(possibly in different time rounds).
There is no central control: every station acts autonomously by means of a distributed algorithm.
We assume that while the parameter $N$ is always 
known to the stations, the actual number $k$ of active stations may be known
or unknown; and we study, among other things, the impact of this knowledge
on the {time complexity}.
\gia{Throughout the paper, when we talk about ``knowledge of a parameter'' $N$ and/or $k$, we mean that a 
\textit{linear upper bound} of $N$ and/or $k$ is known to a station,
\textit{i.e.,} any upper bound of the form $c N + b$
(respectively $c k +b$) for some constants $c\ge 1$ and $b\ge 0$.}

\paragraph{\textbf{Communication.}}
Stations \gia{can transmit and receive messages} on the 
shared channel in synchronous {\em rounds}
(also called {\em time steps} or {\em time slots}).
If $m\leq k$ stations transmit at the same round, then the result of the transmission depends on the
parameter $m$ as follows: 
\begin{itemize}
	\item If $m = 0$, the channel is {\em silent} and 
	no packet is transmitted;
	\item If $m = 1$, the packet owned by the singly transmitting station is {\em successfully transmitted} on the channel and
	therefore 
	delivered
	to all the other stations;
	\item If $m > 1$, simultaneous transmissions interfere with one another (we say that a {\em collision} occurs) and
	as a result no packet is heard by the other stations.
\end{itemize}
No special signal is heard in the case of collision, and therefore it is impossible to distinguish between collision and no transmission. 
This setting is called {\em without collision detection}.\footnote{%
	By contrast, in a {\em collision detection} setting, 
	not considered in this work, the channel elicits a \textit{feedback} (interference noise) in case of collision, 
	allowing to distinguish collision from silence.	
} 
The only feedback a station may receive is when it actually transmits successfully, 
\gia{in which case it gets an {\em acknowledgment}. We distinguish between
two settings: one where acknowledgements are provided to successfully 
transmitting stations (setting \text{with acknowledgment}) 
and one where they are not provided (setting \textit{with no acknowledgement}).}

\paragraph{\textbf{Contention resolution problem.}}
Each of the $k < N$ stations that can become active has a packet that 
can be transmitted in a single time slot. 
The goal is to allow each of these $k$ contending stations to transmit 
successfully its own packet. A {\it contention resolution} 
algorithm is a distributed algorithm that schedules the transmissions 
guaranteeing that every station possessing a packet eventually transmits 
individually (\textit{i.e.}, without causing a packet collision 
with other stations in the same time unit).



\paragraph{\textbf{\gdm{Static vs dynamic scenario.}}}
In this work we focus on the general and realistic situation 
in which each of the $k$ stations 
\gdm{capable of being activated}, can join the channel,
and therefore start its own protocol, in a possibly \textit{different}
time slot (\textit{asynchronous start}). Although the first papers on
contention resolution date back to the 70's, and in
spite of the great attention that the problem has received in {about} 
fifty years of research, it is only recently that the asynchronous start
(also sometimes called \textit{dynamic scenario}) has been 
considered~\cite{Bend-20,Bend-16,CDK,DK,DK_tcs,DS-17}.

Indeed, most of the literature produced so far 
for the multiple-access channel, either assumed that the $k$ stations
are activated simultaneously (\textit{synchronous start}, 
\gdm{also called \textit{static scenario}}) 
~\cite{Cap,CGR,GL,GFL,GW,KG,Kow-PODC-05,AMM13}
or that the activation times are restricted to some statistical model 
(mainly when packet arrivals follow a Poisson distribution) 
or adversarial-queuing models 
\cite{Bend-05,CKR-TALG-12,Goldberg, Kumar, MB,RagUp}. 
%

In the more realistic scenario considered in the present work, 
the stations are totally independent of one another. Consequently, 
each of them can join the channel in a different time, as a new packet becomes 
available and needs to be transmitted. 
We realistically assume that the sequence of activation times, 
also called a {\em wake-up pattern}, is totally determined by a 
\textit{worst-case adversary},
{\textit{i.e.}, the adversary whose goal is to obtain an execution of the algorithm with worst possible measurement (the time measure will be defined later in this section)}.
Clearly, this asynchrony among the activation times 
introduces an additional challenge in designing a contention
resolution algorithm. Indeed, in this setting one has to 
consider that each station can start executing its protocol 
in a different time without knowing, even approximately, the time at which any other
station started its own one. 

Throughout the paper, ``switched on'', ``activated'',
``woken up'' and their derived terms are used interchangeably 
to mean the action, controlled by a worst-case adversary, by which a station wakes up and starts executing the algorithm.

\paragraph{\textbf{Timing.}}
The \gia{transmissions can occur} in synchronous rounds 
(the clocks of all the stations tick at the same rate). 
However, we assume \textit{no global clock} and 
\textit{no system-based synchronization}:
\gia{each station measures time with its own local clock:
it starts counting the time slots when it is activated, so the indices of the time slots of  different stations could be shifted with respect to each other.}
This model, sometimes referred to as \textit{locally synchronous},
must be contrasted with the 
\textit{globally synchronous} model in which all the participating stations 
share a \textit{global clock} and, therefore, 
can synchronize their activities with
respect to the current global round number ticked by the common clock.
Interestingly, this distinction between globally and locally 
synchronous clock, comes into play \textit{only} in the model with 
asynchronous start adopted in the present work.  
Indeed, when all the stations are activated simultaneously, all the local
clocks start together and therefore will tick the same round number.


\paragraph{\textbf{Algorithms.}}\label{algorithms}
This paper focuses on \textit{deterministic non-adaptive
distributed algorithms} for
the contention resolution problem in the model with asynchronous start.
\gia{
Following the standard computational model for distributed processing,
at each time slot a station can execute an arbitrary amount of computation, and 
all computations have to rely on local data only.
Specifically, any station, starting from its activation time has to decide
independently
(\textit{i.e.} without the help of a central control) 
for each time slot, whether to transmit or remain silent. 
Hence, we can identify two distinct tasks/behaviours that have
to be fulfilled in the distributed process: 
there is a \textit{computation} mode aiming at determining the schedule
of transmissions and an \textit{execution} mode which is responsible
of performing them. During the distributed process, the two modes can be temporally scheduled in different
ways depending on  whether the solution is adaptive or non-adaptive.

In {\em non-adaptive} solutions, studied in this paper, 
each station has to produce the entire transmission schedule
beforehand, that is, prior to the execution. In this case,
the entire computation mode has to be completed before the execution
mode starts.
As such, the transmission schedule computed by a station cannot be influenced 
by any information 
gathered during the execution, but can only 
depend on the station ID \dk{and other input parameters (if any)}.
In contrast, in \textit{adaptive} solutions the two modes are interleaved: 
stations are allowed to determine their transmission schedule during 
the execution, so to adapt their transmission behaviour to
the channel feedback and/or any other information gathered while
they are functioning.

The ultimate goal \dk{of non-adaptive algorithms} is to organize the transmission schedules in such a way 
to allow each station to transmit its packet successfully, 
regardless of the activation times of the contending stations. 
}
Among the many advantages of non-adaptive algorithms, we have: fast local 
processing (once the transmission schedules are implemented),
higher resiliency, and independence from collision detection and many other physical resources. 

We consider both algorithms that use acknowledgments to switch off after successful transmission (called {\em acknowledgment-based}), and algorithms which execute 
the whole transmission schedule without switching off 
({\em no-acknowledgment-based}). 

\dk{When specifying input parameters of an algorithm, we separate by a semicolon parameters corresponding to the model setting (i.e., so called ``known parameters'', such as $N$ or $k$) from other parameterization dedicated to a particular algorithm (e.g., constants, integers, transmission sequences). For example, $\mathscr{A}(N,k;c)$ means that algorithm $\mathscr{A}$ knows $N,k$ and its code uses also an external variable $c$.}

\paragraph{\textbf{Performance measures.}}
There are many aspects that can be considered when evaluating the efficiency of a protocol. In this paper we focus on {\textit{time complexity} expressed as \textit{maximum latency}
defined as follows}.

{Consider a single execution of a given algorithm, that is, its run for a fixed wake-up pattern.}
{
From the perspective of a single station, the time efficiency can be expressed
in terms of latency of that station, defined
as the number of rounds necessary for the station to transmit 
its packet successfully, measured since its activation time.}
{From the perspective of the whole execution 
-- the time performance is measured in terms of   
 \textit{maximum execution latency}, 
where the maximum is taken over all activated stations.}

{%
The time performance of the 
algorithm is measured in terms of 
\textit{maximum latency},
defined as the maximum of the maximum execution latencies, over all possible executions (\textit{i.e.}, all possible wake-up patterns).
}

Analogously to the classical Little's Law 
for stochastic queuing systems \cite{little}, we can also evaluate the efficiency of our protocols
in terms of \textit{channel utilization}, defined as the ratio between the
contention size and the maximum latency.
This measure tells us the percentage of slots used for successful 
\gia{transmissions} when the contention on the channel is $k$.




\subsection{Previous work and our contribution}
\label{sec:previous-work}

Contention resolution on a shared channel 
has been studied for decades from various perspectives, including 
communication tasks, scheduling, fault-tolerance, security, energy, 
game-theoretical and many others.
Here we consider only some fundamental results on aspects
that are needed to contextualize our research.

The first theoretical papers on channel contention resolution
date back to more than 40 years ago. The early works, 
inaugurated by the ALOHA algorithm~\cite{Abramson,Roberts}, 
focused on randomized protocols with a collision detection mechanism.
Their performance was analyzed in idealized statistical models in which 
the number of stations was infinite and 
\gdm{packet arrival follows a Poisson process}.  

A new category of protocols (splitting algorithms) was initiated with the
tree algorithm independently found by Capetanakis \cite{Cap}, 
Hayes \cite{Hay}, and Tsybakov and Mikhailov \cite{TM}. 
Although this algorithm was also initially presented in the Poisson statistical model,
it works as a deterministic algorithm for the non-statistical general
situation when only $k$ stations, out of the total
$N$ stations attached to the channel, have packets to be transmitted,
\gdm{provided that the stations start simultaneously (synchronous start)}.

In the worst case, the tree algorithm accomplishes the task in 
$O(k + k\log(N/k))$ rounds, for every $k$ and $N$~\cite{GreenbergPhD,GW}.
The tree algorithm is close to optimal in view of an almost matching 
$\Omega(k\log N/\log k)$ lower bound demonstrated by Greenberg and Winograd \cite{GW}. 
 
Surprisingly, if $k$ (or a linear upper bound on it) is given {\em a priori} to 
the stations, then Koml\'os and Greenberg showed that
the same $O(k + k\log(N/k))$ bound of the tree algorithm 
can be achieved even non-adaptively~\cite{KG} 
in a simple channel with acknowledgments (and, of course, 
without collision detection, which is not usable anyways 
by non-adaptive algorithms). 
\gdm{This is the first paper that shows how acknowledgments can be exploited
to switch off the stations that already transmitted successfully to reduce the 
time complexity of deterministic non-adaptive algorithms.} 
The proof is non-constructive; later Kowalski~\cite{Kow-PODC-05}
showed a more constructive solution, based on selectors~\cite{CGR,Ind}, 
reaching the same asymptotic bound.
Clementi, Monti, and Silvestri \cite{CMS} showed a matching lower bound of 
$\Omega(k \log(N/k))$. 
  
For randomized solutions, an efficient adaptive algorithm using collision detection
was presented by
Greenberg, Flajolet and Ladner \cite{GFL} and Greenberg and Ladner \cite{GL}.
The algorithm works in $2.14 k + O(\log k)$ rounds with high probability 
without any \textit{a priori} knowledge of the number 
$k$ of contenders. 
More recently, Fern\'andez Anta, Mosteiro and Ramon Mu\~noz~\cite{AMM13} 
obtained the same asymptotic (optimal) bound with a non-adaptive algorithm 
also ignoring the contention size $k$.

All the abovementioned results, both for deterministic and
randomized solutions, hold for the static scenario with synchronous start.
While obviously the lower bounds hold also for the general model with
asynchronous start, the upper bounds do not apply as they require synchronization
of the starting points. Therefore, many questions remain open for the
dynamic model with independent asynchronous starting~times.

As far as randomized solutions are concerned, many questions 
have been answered in \cite{DS-17} in the model without collision detection. 
In particular, it has been showed that, in contrast with what happens in
the static model, in the dynamic counterpart there is a separation, 
in terms of {time complexity (maximum latency) and channel utilization}, between non-adaptive algorithms 
ignoring $k$ and algorithms that either are adaptive or know the parameter $k$.  
{
As for adaptive algorithms, efficient solutions have been provided recently
both in terms of throughput and number of transmissions \cite{Bend-20, Bend-16}. 
}

For deterministic solutions very little is known \gdm{in the dynamic scenario}.
Chlebus et al.~\cite{CKR-TALG-12} developed a deterministic non-adaptive 
contention resolution algorithm with latency independent of $k$, essentially
$O(N\log^2 N)$.
Although working for any wake up pattern, \textit{i.e.} 
for any sequence of activation times,
does not provide efficient solutions for $k \ll N$.
When a global clock is available, it has been proved that,
if the stations can switch off after acknowledgements,
the problem can be solved deterministically and non-adaptively 
with a latency of $O(k\log N\log\log N)$ even when $k$
is unknown \cite{DK}. 

\dk{Adaptive deterministic solutions could be easily obtained by combining any known solution to the wake-up problem with any static contention resolution algorithm. In particular, one could apply a wake-up algorithm by Chlebus et al.~\cite{ChlebusGKR05}, which selects a leader in time $O(k \log k \log N)$ from the first awakening, with the contention resolution protocol by Koml\'os and Greenberg~\cite{KG} working with latency $O(k\log(N/k))$. Both these results are existential, and prove an upper bound $O(k \log k \log N)$ on deterministic adaptive dynamic contention resolution. On the other hand, a lower bound 
\gia{ $\Omega(\min\{k^2\log_k N,N\})$ \cite{DR82} }
in the static setting automatically holds in the dynamic one.} 

To the best of our knowledge, the present paper is the
first work on non-adaptive contention resolution algorithms for the
most general model of communication 
with \textit{asynchronous start} and \textit{without global clock}.

\paragraph{\textbf{Our contribution.}}
In this work we explore the impact 
of asynchrony, knowledge of the number of contenders, 
and the availability of acknowledgments, 
on latency for non-adaptive deterministic protocols. 
All our results hold for settings without global clock and 
for any possible sequence of activation times
for up to $k$ competing stations.
\dk{A summary of our results and their comparison with other results obtained in the most relevant settings can be found in Table~\ref{table:results}, while Table~\ref{table:related-results} shows other (a little bit less) related results, which admit fast contention resolution algorithms.}
\gia{For brevity, we use the terms \textit{ack} and \textit{no-ack}
to indicate respectively the case when a station receives an
acknowledgement when it successfully transmits (and so it can
switch off) and the case when no acknowledgement
is received in case of success, and therefore each station remains 
in the system until the end of the execution
(possibly disturbing the other transmissions) even though it  has
already successfully sent its own message.}

We \gia{start showing} that in the acknowledgement-based model,
if the number of contenders $k$ 
is known and each station switches off after receiving the 
acknowledgment of
its successful transmission, the channel admits efficient solutions:
there exists a deterministic non-adaptive distributed algorithm,
called {\known} (\textit{Slow Frequency Increase}), working 
{with maximum latency} $O(k \log k \log N)$.

\begin{table*}[t]
	\small
	\centering 
	\begin{tabular}{| c | c | c | c |}
		\hline
		\textbf{synchrony}    &   \textbf{settings}    & \textbf{upper bound}     &\textbf{lower bound}\\
		\hline
		static                          & $k$ (un)known and no-ack & $O(k^2\log N)$ \cite{CMS}  & \dk{$\Omega(\min\{k^2\log_k N,N\})$} 
  \gia{\cite{DR82}}\\
		& $k$ (un)known and ack & $O(k\log(N/k))$ \cite{KG,Kow-PODC-05}  & $\Omega(k\log(N/k))$ \cite{khasin,CMS}\\
		\hline
		\multirow{3}{*}{dynamic}
		& $k$ known and no-ack      & \bm{$O(k^2\log N)$}  &  
          \gia{ $\Omega(\min\{k^2\log_k N,N\})$ \cite{DR82} }\\
		& $k$ known and ack  & $\bm{O(k\log k\log N)}$ & $\Omega(k\log(N/k))$ \cite{khasin,CMS}\\
		& $k$ unknown and no-ack  &   $\bm{O(k^2\log N)}$  & 
		\gia{ $\Omega(\min\{k^2\log_k N,N\}\dk{+\bm{k^2/\log k}})$ \cite{DR82} }
		\\
		& $k$ unknown and ack &  $\bm{O(k^2\log N/\log k)}$  & $\bm{\Omega(k^2/\log k)}$ \\
		\hline
	\end{tabular}
	\caption{
		Results on latency in {\em non-adaptive deterministic} contention resolution, comparing \dk{our results with those obtained in} the most related settings. Results from this~paper~are shown in bold. Results for static scenarios hold for both known and unknown contention size $k$ (we denoted this by ``$k$~(un)known'').
	}
	\label{table:results} 
\end{table*}

\begin{table*}[t]
	\small
	\centering 
	\begin{tabular}{| c | c | c | c |}
		\hline
		\textbf{\dk{dynamic} model feature}    &   \textbf{settings}    & \textbf{upper bound}     &\textbf{lower bound}\\
		\hline
		      non-adaptive randomized  & $k$ unknown and no-ack & $O\left(k\frac{\ln^2 k}{\ln\ln k}\right)$  \cite{DS-17}  & $\Omega\left(k\frac{\log k}{(\log\log k)^2}\right)$ \cite{DS-17}\\
		    \ & $k$ known and no-ack & $O\left(k\right)$  \cite{DS-17}  & $\Omega\left(k\right)$ \\
		\hline   
		adaptive randomized & $k$ (un)known and ack & $O(k)$ \cite{DS-17}  & $\Omega(k)$\\
		\hline   
		adaptive deterministic & $k$ (un)known and ack & $O(k\log k \log N)$ 
  \cite{ChlebusGKR05,KG} & $\Omega(k\log(N/k))$ \cite{khasin,CMS}\\
		\hline
	\end{tabular}
	\caption{
		\dk{Related results on latency in: {\em dynamic adaptive randomized}, {\em dynamic non-adaptive randomized}, and on {\em dynamic adaptive deterministic} contention resolution. All results hold in the dynamic setting without collision detection (as considered in this work). The lower bound formulas $\Omega(k)$ come directly from the fact that $k$ packets have to be successfully transmitted on a single channel.
      The upper bound $O(k\log k\log n)$ in the last row comes from
      a combination of results in \cite{ChlebusGKR05} and \cite{KG}, see  Section~\ref{sec:previous-work} for more details.}
	}
	\label{table:related-results} 
\end{table*}

\begin{theorem}
\label{t:upper-known}
	There exists some constant $c>0$ and a (deterministic) schedule $\cR$ such
	that algorithm \known$(N,k;c,\cR)$, \gia{in the setting with ack,}
%
\dk{allows 
	any station $v$ to transmit successfully within 
	$O(k\log k\log N)$ rounds, under contention $k$}. 
%
\end{theorem}

This is close to the known lower bound $\Omega(k\log(N/k))$~\cite{CMS}
obtained for adaptive algorithms and in the stronger model with global clock.

In the same settings, if $k$ is unknown,  we show an  
$\Omega(k^2/ \log k)$ lower bound.
\gdm{Under the common assumption that $N$ is polynomially bounded in the number $k$ of active stations (that is $N < k^c$, for a constant $c$) this lower bound}
points out that the ignorance of the contention size $k$ makes the channel \gia{substantially} less efficient, even if the  
stations switch off after acknowledgments. 
Namely, we prove the following result.

\remove{
\begin{theorem}\label{lb:unknown_k}
	No deterministic non-adaptive algorithm 
	achieves an average 
	latency $o(k^2/\log k)$,
	even when the stations can switch off after a successful transmission.
\end{theorem}
}

\begin{theorem}
\label{lb:unknown_k}
	For any $c=1/(4+o(1))$,
	no deterministic non-adaptive algorithm $\mathscr{A}$, \dk{having as an input only parameter $N$,} achieves a maximum 
	latency $t_\mathscr{A}(k) = ck^2/\log k$ \dk{for any contention $k$,}
	even \gia{in the setting with ack, that is} when the 
	stations can switch off after a successful transmission.
\end{theorem}
   
\gdm{For unknown $k$, we also show}   
\gdm{the existence of} an 
algorithm, called \unknown\ 
(\textit{Small Polynomial rate Decrease}), with  $O(k^2 \log N)$ maximum latency, which is achieved even if acknowledgments are not provided
\gia{and so the stations cannot switch-off after a successful transmission}. 
\gdm{For $N$ polynomially bounded in the number $k$ of contending stations, this upper bound nearly matches the lower bound of Theorem \ref{lb:unknown_k}
obtained when the stations \gia{receive an acknowledgement after a successful transmission and therefore can switch off}.}

\begin{theorem}\label{derand}
	There exists some constant $b>0$ and a (deterministic) schedule $\cS$ 
	such that algorithm \unknown $(N;b,\cS)$, \gia{in the
 setting with no-ack,} 
	allows 
	any station $v$ to transmit successfully within 
	$O(k^2\log N)$ rounds, {under contention $k$}. 
\end{theorem}

We finally show that the above algorithm could be further improved 
if stations could switch off upon acknowledgments
\gdm{and under the additional assumption that $\ln\ln N = O(\log k)$}.
{In this case, we prove that there exists an algorithm \unknownack\ (\textit{Small Polynomial Rate Decrease with} {Acknowledgements}) achieving}
a maximum latency of $O(k^2 \log N/\log k)$.

\begin{theorem}\label{derandack}
	There exists some constant $c>0$ and a (deterministic) schedule $\cS$ 
	such that algorithm \unknownack $(N;c,\cS)$,
     \gia{in the setting with ack,} allows  
	any station $v$ to transmit successfully within 
	$O(k^2\log N/\ln k)$ rounds \dk{for contention at most $k$, for any integer $1<k\le N$}, \dk{as long as $\ln\ln N = O(\ln k)$}.
\end{theorem}

\paragraph{\textbf{\gia{Our methodology.}}}
	As in the mentioned seminal work of Koml\'os and Greenberg \cite{KG} for the 
	simplified static model, our \dk{non-adaptive protocols have similar structure (\textit{i.e.}, iterating two loops, for carefully selected parameters) and depend on a schedule computed at the very beginning. Deterministic schedules that make the whole algorithms efficient are non-explicitly-constructive (\textit{i.e.}, no polynomial-time algorithm constructing them is known)} and 
	are obtained using 
	the probabilistic method~\cite{AS}: 
	in order to prove the existence of a transmission schedule 
	with some prescribed combinatorial properties
 \gia{that guarantee a solution to our problem, \textit{i.e.}} a 
 successful
	transmission for each of the contending stations
 \gia{within a certain time},
	we construct a suitable probability space, whose elements are transmission
	schedules, and show that a randomly chosen element in this space has the desired properties with positive probability. 
 \gia{Being the probability strictly larger than zero, 
 a protocol with such properties
 (i.e., a protocol solving the contention resolution problem) must exist.}
 \gia{Analogously, in order to prove our lower bound, we consider
 a probability space whose elements are wake-up patterns and show that
 a randomly chosen element in such a space exhibits the following
 property with positive probability: any non-adaptive deterministic
 algorithm must spend more than $ck^2/\log k$ time slots to get
 the first successful transmission, \dk{for some constant $c>0$.} This implies the existence of
 a wake-up pattern with that property, \textit{i.e.} an input
 instance for the contention resolution problem forcing any 
 non-adaptive deterministic algorithm 
 to have a $ck^2/\log k$ maximum latency.
 }
	
	A transmission schedule, {\em i.e.} 
	a binary sequence specifying for each active station the time slots at which 
	it has to transmit, is a very common notion that can be found in many
	fundamental problems such as wakeup \cite{BKwake, GPP}, 
	broadcast in radio networks \cite{CGR,CMS,GDM}, 
	group testing \cite{du2000combinatorial, PoratR11},
	and, in general, in every selection problem where one or more items, out of an arbitrary ensemble of objects, need to be somehow separated from the others.     
	
	Finding an explicit construction of such schedules 
	turned out to be a very difficult task, so that for many of the above 
	mentioned problems, despite a long history, there are still substantial 
	\gdm{complexity separations}
	between existential bounds and constructive solutions. 
	The interested reader is referred to the book by
	Du and Hwang \cite{du2000combinatorial} and to Indyk's work \cite{Ind}
	to find many examples of such a discrepancy between constructive and
	non-constructive solutions.
	We need only think that for the wakeup problem, which asks for only one station 
	to transmit successfully (as opposed to our problem requiring this task for every active station), the existential bound tells us that a schedule with
	$O(n\log^2 n)$ rounds suffices \cite{GPP}, 
	but the currently best constructive solution is $O(n^{3/2}\log n)$ \cite{BKwake}
	and came after a series of papers with algorithms of increasing efficiency.
	No better ``fully'' constructive 
	solution is known up to date, despite a lot of efforts on such a fundamental problem. 
	Indyk \cite{Ind} showed a schedule of $O(n^{1+\epsilon})$ rounds, for any $\epsilon > 0$,
	but it needs a polylogarithmic number of random bits to be performed
	and therefore cannot be deterministically computed in polynomial time,
	but only in quasi-polynomial time by enumerating all random bits.
	
	Given such difficulties, it is quite
	natural that in many situations it is worth doing a deep investigation 
	on what one could expect from deterministic solutions,
	before even attempting to find an explicit one. 
	Moreover, existential proofs are often theoretically interesting 
	in themselves as they
	can show complexity separations between different settings, and also because
	they could reveal some important clue on the explicit algorithm that could 
	inspire for possible constructions.  
    In this perspective, we believe that our results have a number of 
    important conceptual implications that can be outlined as follows.

First, they imply that the knowledge of the contention size has
an important impact on time complexity 
as it is shown by the upper bound for $k$ known (Theorem \ref{t:upper-known}) 
and the lower bound for $k$ unknown (Theorem \ref{lb:unknown_k}). 
\gdm{Indeed, when $N$ is polynomially bounded in $k$, they imply 
	a nearly quadratic separation on latency between 
	the cases $k$ known and $k$ unknown.
This in turn implies an exponential complexity separation
on channel utilization, 
defined as the ratio between the
contention size and the maximum latency. Indeed,
Theorem \ref{t:upper-known} implies that when $k$ is known there is an 
$\Omega(1/(\log k\log N))$ lower bound
on channel utilization, while for $k$ unknown Theorem \ref{lb:unknown_k}
implies an $O(\log k/ k)$ upper bound. 
}
This is a remarkable fact as it is known that for the static model with
synchronous start the knowledge of the contention size does not influence
asymptotically the latency (and the channel utilization), 
as follows from the matching upper bound by Koml\'os and Greenberg \cite{KG} and
the lower bound by Clementi, Monti and Silvestri \cite{CMS}.


The second implication concerns the impact of acknowledgments:
they quadratically improve the latency
if (some estimate of) $k$ is known, as it can be inferred from our
upper bound when $k$ is known and acknowledgments are available (Theorem \ref{t:upper-known}) and the 
$\Omega((k^2\log N)/\log k)$ lower bound obtained in the context of 
superimposed codes~\cite{Chau96}.
This translates into an exponential improvement on channel utilization,
\gdm{namely we have an $\Omega((k^2\log N)/\log k)$ lower bound with 
	acknowledgments and
an $O(\log k/ (k\log N) )$ upper bound without 
acknowledgments.}

This is also notable in view of the fact that for non-adaptive randomized algorithms 
it is known that acknowledgments are not particularly helpful, as
$\tilde{O}(k)$ time slots are sufficient in absence of acknowledgments 
(even for $k$ unknown)~\cite{DS-17}. 
\gdm{In terms of channel utilization, this implies that 
the improvement is only at most polynomial in this case.}

\paragraph{\textbf{Structure of the paper.}}
In Section~\ref{s:known-k}
we provide a non-adaptive deterministic solution (making use of acknowledgments) with latency 
$O(k\log k\log N)$. The lower bound $\Omega(k^2/\log k)$ on latency when the contention
is unknown is given in Section~\ref{s:lower}. 
The almost matching solutions with latency $O(k^2 \log N)$
and $O(k^2 \log N/\log k)$, 
without using acknowledgments and with acknowledgments respectively, are described
in Section~\ref{s:unknown-k}.
We conclude with Section~\ref{s:conclusions} by discussing various implications
 and open directions.

Throughout the paper, for the sake of simplicity, 
we will ignore rounding to integers as this does not affect
the asymptotic bounds in our proofs: \textit{e.g.} we will assume that $\log N$
and $\log k$ are integers. The same simplification will be used 
on fractions of stations like $k/3$ or time rounds like $k/\log k$ whenever this does not cause ambiguity.

\section{Algorithm for known contention size \gia{and acknowledgements}}
\label{s:known-k}

In this section we assume that the number $k$ of contenders 
(or some linear upper bound on it) is known. Moreover, it is supposed that the
stations can use acknowledgments to switch off after a successful 
transmission.  

In this setting, we prove that it is possible to design a deterministic distributed algorithm solving the contention with latency $O(k\log k\log N)$.
\gia{
The algorithm is non-adaptive: all decisions regarding when to transmit 
and when to stay silent are made before the station starts transmitting.
Namely, recalling the definition given in paragraph ``Algorithms''
on page \pageref{algorithms}, the transmission schedule
$\cR$, which specifies for each active station when to transmit and when to stay silent, is entirely created in the computation mode, 
\textit{i.e.} \dk{in the very beginning},
and can no longer be changed during the execution mode.
Recall that, according to the standard distributed computational model, 
at any time slot a station can execute an arbitrary amount
of computation, so the transmission schedule can be
computed in the same slot when the station is activated.
}
The transmission schedule $\cR$ is composed of $N$ arrays $\cR_v$, one for 
each station $v$. Each transmission array $\cR_v$ is a sequence of bits
corresponding to the time slots of $v$'s local clock: the station transmits in the $r$th slot of its local time if the $r$th bit of $\cR_v$ is 1 and stays silent otherwise.

We proceed by showing first that a specific random instantiation of 0-1
arrays $\cR_v$, for each $v$, allows our algorithm to solve 
the contention with latency $O(k\log k\log N)$ 
with very high probability;
then, we derandomize the array to obtain the desired deterministic solution.
All of our arguments are meant to hold for $k$ and $N$ sufficiently large.

Our contention resolution algorithm when the contention size is known 
is formally described by the following protocol
executed by any station $v$ starting from the time at which it is activated. 
Apart from parameters $N$ and $k$, it takes as input a constant parameter $c$, whose
value will be specified later in the analysis, 
and the transmission schedule $\cR$ that will be specified shortly. 
Each array $\cR_v$ is synchronized with $v$'s local clock in the
sense that $v$ starts reading the array at the activation time and 
at the $r$th round of its clock, $v$ reads the $r$th bit of $\cR_v$
and transmits if and only if this bit is 1.
 
\gdm{The protocol's execution mode is organized
in $2\log k+1$ consecutive phases, 
each one lasting $T = ck\log N$ time rounds.
For $j = 1,2,\ldots,T$,
the $j$th round of phase $i$, for $i = 0,1,\ldots, 2\log k$, 
will correspond to round}
number $i\cdot T+j$ of $v$'s local clock. The entry of $\cR_v$ 
corresponding to this round will be denoted $\cR_{v,i,j}$. 
Hence, station $v$ will transmit in the $j$th round of phase $i$ if and only if $\cR_{v,i,j} = 1$. 
{
Our protocol is called \known (\textit{Slow Frequency Increase}) as 
its main feature  is a
slow increase of the frequency of transmissions, which remains constant 
during each phase and is less than doubled from one phase to the next one (see Definition~\ref{def:sched}).}

\begin{algorithm}[ht]\label{prt:known}
	\caption{\known$(N,k;c)$: executed by station $v$}
	\label{alg:suniform}
	\begin{algorithmic}[1]

        \State {\gia{Compute $\cR_v$}}
 
		\For{$i\gets 0,1,2,3, \ldots,  2\log k$}
		\For{$j\gets 1,2,\ldots,T=ck\log N$} 
		\State at round $i\cdot T+j$ of its local clock $v$ transmits if and only if $\cR_{v,i,j}=1$
		\EndFor
		\EndFor
	\end{algorithmic}
\end{algorithm}

{
The protocol execution for any station $v$ starts when $v$ 
is activated and lasts
until it successfully transmits (getting an acknowledgment),
in which case it immediately switches off.
Essentially, we need only to show that any activated station 
will transmit successfully before the termination of its protocol 
execution -- 
then the algorithm's latency follows 
immediately from the maximum duration of the protocol
execution, which is $T (2\log k+1) = O(k\log k\log N)$ 
rounds.
}

\gst{Let us call \textit{meaningful rounds} the $(2\log k + 1)T$ rounds that follow the activation time of a station.
An interval of consecutive rounds, each of them being meaningful, is called a 
\textit{block}.}
Since at most $k$ stations can be activated during the execution of the algorithm, 
the following fact holds.

\begin{fact}\label{fact:kt}
	Any execution of the algorithm can be partitioned into at most $k$ 
	disjoint blocks, each of them having length at most 
	$k\cdot T (2\log k+1)$.
\end{fact}


Let us number $1,2,3,\ldots$ the subsequent rounds of any block.
We can partition these rounds into \textit{segments}, 
each  of them consisting of $T$ subsequent rounds.
Thus, segment 1 has rounds $1,2,\ldots,T$, segment 2 has rounds $T+1,T+2,\ldots,2T$, and so on. Altogether there are at most 
$k(2\log k+1)$ segments in a block.

\gdm{
Recall that in our model there is no global clock, so the  
round number of a block is a global round number that will not be visible 
to the stations, but will only be used for the purpose of the analysis. 
In fact, it is worth noting that a segment is a sequence of $T$ consecutive 
rounds of such an invisible global measuring, while the previously defined 
phase is a sequence of $T$ consecutive rounds of a single station's local clock. 
To avoid confusion, we will use variable $t$ for the rounds of a block, while 
indices $i,j$ will be used to refer to the round number of some station's local
clock (specifically the $j$th round of the $i$th phase of the protocol execution
of this station). Observe that, as a consequence of the asynchrony among 
activation times, at any round $t$ of a block, different stations can be involved in different rounds $j$ and phases $i$ of their respective protocol executions. 
}

We can observe that any station is activated within a block and remains active 
only in this block. Consequently, the set of activated stations can be partitioned 
in as many subsets as the number of blocks.


\subsection{Analysis of algorithm \known\ for a random~$\cR$}
\label{s:known-random}

As announced earlier, we first analyze our algorithm for a
randomly instantiated schedule $\cR$.
Namely, we consider a random schedule defined as follows.

\begin{definition}(Random schedule $\cR$)\label{def:sched}
	Let $\cR$ be formed by the following sequence of independent Bernoulli trials:
	for any station $v$ and indices $i,j$, let $R_{v,i,j}=1$ with probability 
	$p = \min\{\frac{1}{2},\frac{{2}^{i/2}}{2k}\}$, and $R_{v,i,j} = 0$ with
	probability $1-p$.
\end{definition}

For any fixed block, the 
\textit{sum of transmission probabilities} at 
a given round $t$ (of the block) is defined as 
\[
\sigma(t)=\sum_{v\in V(t)} p_v(t),
\]
where $V(t)$ is the set of active stations at round $t$ and $p_v(t)$ is the probability of transmission of station $v$ in round $t$.


We start with the following simple lemma that gives us a lower bound on the 
probability that a given station transmits successfully when the sum of 
transmission probabilities is small.

\begin{lemma}\label{lemma:cici}
	If for a round $t$ of a block $\sigma(t) \le 2$, then the probability that a given station $v$ transmits successfully in round $t$ is larger than $p_{v}(t)/16$.
\end{lemma}

\begin{proof}
	The random schedule $\cR$ in the input of the protocol is such that 
	$p_{v}(t) = \min\left\{\frac{1}{2},\frac{{2}^{i/2}}{2k}\right\} \le 1/2$  
   (see Definition~\ref{def:sched}),
	for any round $t$ and station $v$. 
	Therefore, the probability that a given station $v$ transmits successfully in round $t$ is:
	\begin{eqnarray*}
		p_{v}(t) \cdot \prod_{w\neq v} (1-p_{w}(t)) &=&    p_{v}(t) \cdot \prod_{w\neq v} (1-p_{w}(t))^{\frac{1}{p_w(t)}p_{w}(t)}\\ 
		&>&    p_{v}(t)\cdot (1/4)^{\sigma(t)} 
		\ge p_{v}(t)/16 \ ,
	\end{eqnarray*}
	\gia{where in the second-to-last inequality we have used 
 $p_w(t) \le 1/2$,} while the
 last inequality follows from the assumption that $\sigma(t) \le 2$.
\end{proof}


%

Let ${\cal A}_{\ell}$ be the event that in some round $t$ of segment 
$\ell$ of a block the sum $\sigma(t) < 1/2$. 
The following lemma establishes a relationship between the sum of
transmission probabilities
for two subsequent segments of a block. 

\begin{lemma}\label{fact:aj}
	If for some segment $\ell \ge 1$ of a block event ${\cal A}_{\ell}$ occurs, 
	then in all rounds 
	$t$ of the next segment $\ell+1$, we have $\sigma(t)\leq 2$.
\end{lemma} 
\begin{proof}
	Since the transmission probability of each station doubles every $2T$ rounds,
	and each segment lasts $T$ rounds, we have that from segment $\ell$ to segment 
	$\ell+1$ 
	the stations can at most double their probabilities. By hypothesis
	we know that $\sigma(t) < 1/2$ from some round $t$ in segment $\ell$, therefore
	the sum of transmission probabilities of stations activated not later than 
	$t$, cannot be larger than $2\cdot(1/2) = 1$.	
	
	Moreover, possibly new activated stations after round $t$
	can cause an additional increase of less than 1 unit to the sum, 
	since new stations are no more than $k$ and start sending with probability 
	$1/(2k)$ 
	\gdm{in the first $T$ rounds and $1/(\sqrt{2}k)$ in the following T rounds}.
\end{proof}

The main purpose of this subsection is to prove that our algorithm fails, 
on a random schedule $\cR$, with probability exponentially small in $k$,
namely smaller than $(Nk^2T)^{-k}$. 
First we focus on the probability of failing in a fixed block 
(Lemma~\ref{l:known-random}) 
and then, in Theorem \ref{t:known-random}, by taking the union bound over
all possible blocks, we will derive the final probability bound.

Let $h \le k\cdot (2\log k+1)$ be the number of segments of a block. 
We start with the following result, preparatory to Lemma~\ref{l:known-random}, 
which shows that all events ${\cal A}_1, {\cal A}_2, \ldots, {\cal A}_h$ hold simultaneously with high probability.
This intersection ${\cal A}_1\cap {\cal A}_2\cap \cdots \cap {\cal A}_h$ of events does not occur 
if and only if event ${\cal A}_{\ell}$ 
does not occur for some $1\le \ell \le h$.
The probability that $\ell$ is the smallest index for which ${\cal A}_{\ell}$ 
does not occur is
not larger than $\Pr(\overline {{\cal A}_{\ell}}|{\cal A}_1 \cap {\cal A}_2 \cap \cdots \cap {\cal A}_{\ell-1})$.

\dk{To simplify the presentation in this section, let us define a temporary variable} $P = k^{-3}(Nk^2T)^{-k}$.
We recall that all the following results are meant to hold for $k \leq N$ with $N$ and $k$ being sufficiently large.

\begin{lemma}\label{lemma:neg}
	For a sufficiently large constant $c$ we have that 
	for any segment 
	$\ell \ge 1$ of a block, 
	\[
	\Pr\left(\overline {{\cal A}_{\ell}} \middle|  \bigcap_{i<\ell}  {\cal A}_i\right)\le P.
	\]

\end{lemma}

\begin{proof}
	We have to prove that if the intersection 
	${\cal A}_1 \cap {\cal A}_2 \cap \ldots \cap {\cal A}_{\ell-1}$ occurs, then 
	the event that $\sigma(t) \ge 1/2$ 
	in every round $t$
	of the $\ell$-th segment, occurs 
	with probability at most~$P$.
	\gst{It will be easy to see that for the first segment
	this probability is obtained unconditionally.}
	
	Fix any $\ell\ge 1$. 
	The $\ell$th segment of any block goes from round $(\ell-1)T+1$ to round $\ell T$
	of the block. For any execution of our algorithm,
	we now define a random binary sequence 
	$\rho = \rho_{(\ell-1)T+1},\rho_{(\ell-1)T+2},\ldots, \rho_{\ell T}$, 
	associated with segment $\ell$. Bit $\rho_t$, for $(\ell-1)T+1 \le t \le \ell T$, 
	 corresponds to round $t$ of the block. 
	Since we are assuming 
	that ${\cal A}_1 \cap {\cal A}_2 \cap \ldots \cap {\cal A}_{\ell-1}$ occurs, it follows by Lemma~\ref{fact:aj} that $\sigma(t) \leq 2$ for all $T$ rounds of segment $\ell$. 
	Consequently, for each round $t$ inside the $\ell$th segment, 
	\gdm{two cases can arise:
	either $1/2\leq\sigma(t)\leq 2$ or $\sigma(t)< 1/2$.
	The following procedure 
	sets the value of $\rho_t$ according to different random choices
	depending on which of the two cases above arises.
	The aim will be to show that the probability that the procedure
	never assigns a $\rho_t$ according to the random choice made when 
    case $\sigma(t)< 1/2$ arises, that is when event ${\cal A}_{\ell}$
    occurs, is at most $P$.
    } 
	Let $q_t$ be the probability of having a successful transmission in round $t$.
	All the following random choices are made independently.
	
	\begin{enumerate}
		\item If $1/2\leq\sigma(t)\leq 2$:
		\begin{enumerate}
			\item if there is no successful transmission in round $t$, we set $\rho_t=0$;
			
			\item if there is a successful transmission in round $t$, we set 
			$$
			        \rho_t = \begin{cases}
				             1,  \text{with probability}\ \frac{1}{8q_t}; \\
				             0,  \text{with probability}\ 1-\frac{1}{8q_t}.
			        \end{cases}
			$$
		\end{enumerate}
		\item If $\sigma(t)< 1/2$, then we set 
	    $$
		    \rho_t = \begin{cases}
		             1, & \text{with probability}\ \frac{1}{8}; \\
		             0, & \text{with probability}\ 1-\frac{1}{8}.
		\end{cases}
		$$
	\end{enumerate}
    Note that $q_t$ is defined by the transmission probabilities of the protocol:
    \[
    q_t = \sum_{v} p_{v}(t) \cdot \prod_{w\neq v} (1-p_{w}(t)). 
    \]
    Therefore, the probability of having a 1 in the sequence is well defined. In case 2, this is $1/8$.
    In case 1, $\rho_t = 1$ with probability $q_t(1/(8q_t)) = 1/8$. 
	Hence, for $(\ell-1)T \le t \le \ell T$, 
	each bit $\rho_t = 1$ with
	probability $1/8$, independently of the values of other bits in the sequence $\rho$.

	Let us call \textit{artificial} all the bits $\rho_t$ produced 
	in case 2 of the random assignment above, 
	\textit{i.e.}, when $\sigma(t) < 1/2$. 
	In order to prove the lemma, it suffices to show that the probability that there is no artificial bit in $\rho$ (which is equivalent to event $\overline {{\cal A}_{\ell}}$)  is \gdm{at most} $P$.

	Let $X$ be the random variable defined by the number of 1's in 
	$\rho$. 
	Each time a bit is set to 1 in case 1, 
	there is some station that successfully transmits.
	The number of such successful transmissions cannot
	exceed $k$, as we have at most $k$ active stations and each of them can 
	transmit successfully only once (after which it switches off).
	Therefore, 
	if $X > k$ there must be some artificial bit in $\rho$, which implies that ${\cal A}_{\ell}$ occurs. 
 \gia{Hence, we can \dk{upper} bound the probability of $\overline {{\cal A}_{\ell}}$
 \dk{by}
 the probability that $X\le k$}.
 Therefore, to conclude the proof, it remains to show that $\Pr(X\leq k)\leq P$.
	
	The expected number of 1's in $\rho$ is $\mu = \E X = T/8$.	
	By the Chernoff bound, $\Pr(X\leq (1-\delta)\mu)\leq\exp\left(-\frac{\delta^2\mu}{2}\right)$, for any $0<\delta<1$.
	Since $T = ck\log N$, for a sufficiently large $c$ we can write $k \le T/16$.
	Therefore, 
	\begin{eqnarray*}
		\Pr(X\leq k) &\leq & \Pr(X\leq T/16) \\
		&=    & \Pr(X\leq (1-1/2)\cdot T/8)\\
		&\leq & e^{-\mu/8} = e^{-T/64}=e^{-(c/64)k\log N}
		\ .   
	\end{eqnarray*}
	Assuming that $N$ and $k$ are sufficiently large and $k\leq N$, 
	the lemma follows by
	observing that for a sufficiently large constant $c$,  
	\[
	e^{-(c/64)k\log N} \le k^{-3}(Nk^2T)^{-k} = P.
	\]
\end{proof}

Now we can estimate the probability of having a failure inside a block.

\begin{lemma}
	\label{l:known-random}
	For any fixed block, the probability that there is a station that is active 
	in this block but fails to transmit successfully, is smaller than $k^2 P$.
\end{lemma}

\begin{proof}
	Fix a block.
	Let ${\cal B}_v$ be the event that a station $v$ (active in this block) 
	does not transmit successfully
	in any of the rounds of the block.
	The probability defined in the statement of the lemma can be bounded as follows:
	\begin{eqnarray}
	\lefteqn{\Pr(\exists v:\; {\cal B}_v)} \nonumber \\ 
	&\le &  \sum_v\Pr({\cal B}_v|{\cal A}_1 \cap {\cal A}_2 \cap \cdots \cap {\cal A}_h) + 
	  \Pr( \overline{ {\cal A}_1 \cap {\cal A}_2 \cap \cdots \cap {\cal A}_h}) \ , \label{eq:Bv}
	\end{eqnarray}
	where $h \le k(2\log k+1)$ is the number of segments of the block.
	
	Let us start with the estimation of $\Pr({\cal B}_v|{\cal A}_1 \cap {\cal A}_2 \cap \cdots \cap {\cal A}_h)$, 
	\textit{i.e.}, the probability that a given station $v$ does not manage 
	to transmit successfully when all events ${\cal A}_{\ell}$, for $1\leq \ell \leq h$, occur.
	Let us consider the last $T$ rounds of the protocol execution for $v$.
	These are the $T$ rounds of the last phase $i = 2\log k$.
	\gdm{In each of these rounds}, $v$ transmits (if it hasn't transmitted successfully yet) 
	with probability  
	$\min\{\frac{1}{2},\frac{{2}^{i/2}}{2k}\}$ (see 
    Definition~\ref{def:sched}), which is at least  
	$1/2$ for $i = 2\log k$.
	Moreover, the occurrence of event ${\cal A}_1 \cap {\cal A}_2 \cap \ldots \cap {\cal A}_h$ 
	guarantees (by Lemma~\ref{fact:aj})
	that we must have 
	$\sigma(t) \le 2$ for every round $t$ of these last $T$ rounds.
	By Lemma~\ref{lemma:cici}, it follows that the probability that $v$ transmits successfully in any of these 
	$T$ rounds is larger than $\frac{1}{2}\cdot \frac{1}{16} = \frac{1}{32}$.
	Thus,
	\begin{eqnarray*}
		\lefteqn{\Pr({\cal B}_v|{\cal A}_1 \cap {\cal A}_2 \cap \cdots \cap {\cal A}_h)}\\ 
		& \leq & (31/32)^T  
		=   (31/32)^{ck\log N}
		=  N^{-ck \log(32/31)}\ .
	\end{eqnarray*}
	Assuming that $N$ and $k$ are sufficiently large and $k\leq N$,
	for a sufficiently large constant $c$ we have 
	$N^{-ck \log(32/31)}
	    \leq k^{-3}(Nk^2T)^{-k} = P$.
	
	Now, resuming from (\ref{eq:Bv}) we are ready to estimate the 
	probability of a failure in the fixed block:
	\begin{eqnarray*}
		\Pr(\exists v:\; {\cal B}_v) &\leq &  k P + \Pr\left(\overline {{\cal A}_1 \cap {\cal A}_2 \cap \cdots \cap {\cal A}_h}\right)\\
		&\leq &  k P + \Pr(\overline{{\cal A}_{1}}) + \sum_{\ell = 2}^{h}\Pr(\overline{{\cal A}_{\ell}}|{\cal A}_1 \cap {\cal A}_2 \cap \cdots \cap {\cal A}_{\ell-1})\\
		\mbox{(Lemma~\ref{lemma:neg})}&\leq &  k P +  \sum_{\ell = 1}^{h}P
		\leq   k^2 P \ ,
	\end{eqnarray*}
	where the last inequality holds because $h \le k(2\log k+1)$ and $k$ is sufficiently large.
\end{proof}

The following theorem shows that algorithm \known\ instantiated by the 
random schedule $\cR$ fails with probability exponentially small in $k$.

\begin{theorem}
	\label{t:known-random}
	The latency of algorithm \known\ is $O(k\log k \log N)$,
	with probability larger than $1-k^3P$.
\end{theorem}

\begin{proof}
\gdm{	
	Lemma~\ref{l:known-random} implies that 
	for any given block, the probability that an active 
	station does not manage to transmit successfully during
	this block, is smaller than $k^2 P$. 
	Applying the union bound over all the blocks,
	that are at most $k$,
	it follows that the probability that any station 
	could fail to transmit successfully (within the block in which it is active) is smaller than 
	$k\cdot k^2 P = k^3 P$.
	Now the theorem follows by the simple 
		observation that each station is active for at most 
		$T(2\log k+1) = O(k\log k \log N)$ 
		consecutive rounds.
	}
\end{proof}

\subsection{Analysis of algorithm \known\ for specific deterministic $\cR$}

In this subsection we finally give our theorem on deterministic algorithms when
$k$ is known, which 
represents the main result of the entire section. 
{In order to prove Theorem~\ref{t:upper-known},
we derandomize the results obtained in 
the previous subsection, Theorem~\ref{t:known-random},} by considering all possible input configurations of up to $k$ stations.
\gdm{In particular, this implies that although the result of 
	Theorem~\ref{t:known-random} holds
	for a randomized algorithm working against an oblivious
	adversary, \textit{i.e.} with respect to a fixed arrival pattern for the stations, 
	the deterministic solution of Theorem~\ref{t:upper-known} will work against 
	any kind of adversary.
}  

\remove{
\begin{theorem}
	There exists some constant $c>0$ and a (deterministic) schedule $\cR$ such
	that algorithm \known$(N,k;c,\cR)$ accomplishes contention resolution 
	with latency $O(k\log k\log N)$ (and channel utilization
    $\Omega(1/(\log k\log N))$).
\end{theorem}
}

\begin{proof}[Proof of Theorem~\ref{t:upper-known}]
	Consider a randomly generated schedule $\cR$ as defined in  Subsection~\ref{s:known-random} (see Definition~\ref{def:sched}).
	By Theorem~\ref{t:known-random} there exists a constant $c$ such that
	the probability that our algorithm \known$(N,k;c,\cR)$ 
	does not allow any station $v$ to transmit successfully within 
	$O(k\log k \log N)$  
	is smaller than 
	\begin{equation}\label{P}
      k^3 P = (Nk^2T)^{-k}.
	\end{equation}
		
	Each station can be activated in any time slot within one of the possible $k$ blocks
	in which any execution of the algorithm can be partitioned (see
	Fact~\ref{fact:kt}). Since the length of any block is at most 
	$k\cdot T(2\lceil \log k\rceil+1)$, it follows that each station,
relatively to its activation time, can take a total of	
	$k^2\cdot T(2\lceil \log k\rceil+1) +1$ configurations:
one corresponds to the case the station is not activated 
and the other $k^2\cdot T(2\lceil \log k\rceil+1)$ correspond to the possible activation times.

	Furthermore, there are $\binom{N}{k}$ possible ways of choosing the $k$ stations out of an
	ensemble of size $N$.  
	Hence, the total number of wake-up patterns (possible ways in which the activation 
	times of up to $k$ stations can be arranged), is at most
    \begin{eqnarray*}
        \binom{N}{k}\left(k^2\cdot T(2\lceil \log k\rceil+1)+1\right)^k 
          &<& \left(\frac{N\cdot e}{k}\right)^k \left(2k^2\cdot T(2\lceil \log k\rceil+1)\right)^k \\
          &<& \left( Nk^2T \right)^k,
    \end{eqnarray*}
	where the last inequality holds for sufficiently large $k$.
	
	Resuming from (\ref{P}), the probability that there exists a wake-up pattern on which our algorithm
    fails, is therefore smaller than $(Nk^2T)^{-k} (Nk^2T)^k = 1$.
    This implies that there must exist some schedule $\cR$ such that
    all active stations transmit successfully in any execution of the algorithm
    (\textit{i.e.}, for any choice of the active stations and any wake-up pattern).

    As in the proof of Theorem~\ref{t:known-random}, the bound
    on latency  can be easily derived
    by the simple observation that each station is active for at most 
    $T(2\log k+1) = O(k\log k \log N)$ 
    consecutive rounds.
\end{proof}

\section{Lower bound for unknown contention size}
\label{s:lower}

In this section we prove Theorem~\ref{lb:unknown_k}: the ignorance of the number $k$
of stations (or any linear upper bound on it), 
has a considerable impact on the time complexity of non-adaptive
deterministic contention resolution algorithms, even when the
stations are allowed to switch off after a successful transmission.
 
An arbitrary deterministic non-adaptive algorithm
 can be characterized by the following assumptions. 
Each station has a unique ID in $[N] = \{1,2,\ldots , N\}$.
When a station \gdm{is awakened}, its local clock starts
and its ID determines in which rounds it transmits. Specifically,
the transmission schedule of a station is generated in advance and 
depends only on the ID of the station.

An \textit{input instance}  \gdm{specifies a subset of $k$ ID's 
	and a wake-up pattern for them}, 
\textit{i.e.} the time at which each \gdm{ of the specified $k$ stations} is activated.
In the following analysis, our aim will be to build a worst-case 
input instance for any deterministic non-adaptive algorithm. 

For any fixed input instance of $k$ stations and an algorithm $\mathscr{A}$, 
let us denote by $t_\mathscr{A}(k)$ the maximum latency of $\mathscr{A}$ among all 
activated stations in the input instance. 
We will prove that for any algorithm $\mathscr{A}$, $t_\mathscr{A}(k)$ 
is $\Omega(k^2/\log k)$, or more precisely, we will show
that \textit{no algorithm} $\mathscr{A}$ can achieve a maximum latency 
$t_\mathscr{A}(k) = \frac{k^2}{(4+o(1))\log k}$.

%
%
\remove{
\begin{theorem}
	For any $c=1/(4+o(1))$ there exists an input instance of $k$ stations on which
	no deterministic non-adaptive algorithm $\mathscr{A}$ achieves a maximum 
	latency $t_\mathscr{A}(k) = ck^2/\log k$,
	even when the stations can switch off after a successful transmission.
{[[[I thought it should be: for any algorithm there is an instance ...]]]}
\end{theorem}
}
Our proof will be by contradiction:
we fix a non-adaptive deterministic contention resolution 
algorithm $\mathscr{A}$ that for every sufficiently large $k$ 
 has maximum latency $t_\mathscr{A}(k) = \frac{k^2}{(4+\delta)\log k}$ 
 on every instance of $k$ stations, {where $\delta>0$ could be arbitrarily small.
 We denote $c=\frac{1}{4+\delta}$.}
 

All the following results are meant to hold for any $k$ sufficiently large.
We start with a simple lemma on the least number of transmissions that
algorithm $\mathscr{A}$ assigns to any station in 
the first $t_\mathscr{A}(k)$ rounds of its activity. 
It is important to point out that this is the number of transmissions
assigned to the station by the algorithm, not necessarily the actual number of
times that the station transmits, as it may switch off 
{if its successful transmission takes place during the considered time period (i.e., the actual number of transmissions depends on an execution of the whole channel: input instance and transmission schedules of other stations).}

\begin{lemma}
	\label{bound:sk}
	Each station is assigned at least $k$ transmissions within
	the first $t_\mathscr{A}(k)$ rounds from its activation.
\end{lemma}

\begin{proof} 
	Algorithm $\mathscr{A}$ will assign a possibly different number of transmissions,
	within its first $t_\mathscr{A}(k)$ rounds of activity, to each different station.
	Let $s_k$ be the smallest of these numbers. We want to prove that $s_k \ge k$.
	Assume, to the contrary, that $s_k\leq k-1$ and that this number 
	of transmissions is assigned to some station $v$.
	Let
	$1\le r_1 < r_2 < \ldots < r_{s_k} \le t_\mathscr{A}(k)$ be the rounds, with
	respect to $v$'s local clock, at which $v$ transmits.

	Let us consider an input instance in which a total of $s_k +1 \le k$ stations
	are activated. We now define the wake-up patterns for $v$ 
	and the other $s_k$ stations in such a way that all sending attempts of
	$v$ will collide with some other station's transmission.
	(Remember that all transmission schedules for these $s_k+1$ stations are fixed 
	in advance by algorithm $\mathscr{A}$ independently of the ways in which 
	the activation times can be chosen.)
	We let station $v$ be the first to be activated.
	The activation times for the other $s_k$ stations are arranged in such a way
	that the $i$th station, for $1\le i\le s_k$, has its first transmission in
	round $r_i$ of $v$'s local clock.

	In this way, $v$  is not able to transmit successfully 
	within its first $t_\mathscr{A}(k)$ rounds, which contradicts the initial assumption
	that $\mathscr{A}$ is a contention resolution algorithm with latency $t_\mathscr{A}(k)$.
\end{proof}

The following relatively straightforward lemma provides an upper bound on the
probability that at most one event, out of $m$, occurs, expressed in terms of 
the sum of probabilities of the single events.

\begin{lemma}\label{bound:prob}
	Let ${\cal E}_1,{\cal E}_2,\ldots, {\cal E}_m$ be $m$ \gia{mutually independent} events,
	for any $m>0$,
	occurring with probabilities $p_1,p_2,\ldots,p_m$ respectively and 
	\gdm{$p_1+p_2+\cdots+p_m \ge s$, \gia{for some $s \ge 1$}}.
	The probability $P$ that at most one of these events occurs is at most $(s+1)e^{1-s}$. 
\end{lemma}

\begin{proof}
		We have:
	\begin{eqnarray*}
		P &=    & \prod_j (1-p_j)+ \sum_j p_j\prod_{i\neq j}(1-p_{i})\\
		&\leq &  e^{-(p_1+p_2+\cdots+p_m)}+\sum_j p_j e^{-(p_1+p_2+\cdots+p_m)+1}\\
        &\le &  \gia{e^{-s} + e\cdot\sum_j p_j e^{-\sum_j p_j}}\\
		&\leq & \gia{e^{-s} + e\cdot s e^{-s}}
        \leq (s+1)e^{1-s},
	\end{eqnarray*}
\gia{ where the second to last inequality is due to the fact that
 $s e^ {-s}$ is decreasing for $s \ge 1$.}
\end{proof}


\gdm{In the next lemmas we will use the probabilistic method 
to show the existence of a wake-up pattern for $k$ stations such that
none of them successfully transmits in the interval $[1, \Omega(k^2/\log k)]$,
where round 1 corresponds to the first wake-up time of the $k$ stations.
More precisely, we show that there exists a wakeup schedule 
for $k$ stations such that, 
	conventionally assuming a clock starting at the first wake-up time,}
no successful transmission occurs in rounds $i\in[1,ck^2/\log k]$ for a 
sufficiently large $k$, where we recall $c=1/(4+\delta)$.
The existence of such a schedule implies that $t_\mathscr{A}(k)>ck^2/\log k$
which contradicts the assumption that $t_\mathscr{A}(k)\le ck^2/\log k$ and ends the proof of  Theorem~\ref{lb:unknown_k}.

\begin{lemma}\label{interval}
Assume that no station is switched off.
Let $0<a<b\leq 1$ and $t_\mathscr{A}(\kappa)\le c\kappa^2/\log\kappa$ for any $\kappa\ge ka$.
Let us distribute randomly and uniformly {the wake-up times of}
$x$ stations in 
the interval $[1,ck^2b^2/\log (bk)]$.
In any round 
$t\in [ck^2a^2/\log (ak),ck^2b^2/\log (bk)]$,
the sum of {transmission probabilities} $\sigma(t)$ is at least
$ax\log(bk)/(ckb^2)$.
\end{lemma}

\begin{proof}
Since 
$t_\mathscr{A}(ak)\le c(ak)^2/\log (ak)$,
   {by Lemma~\ref{bound:sk}}
there are at least $ak$ scheduled transmissions
of any station during its first $a^2 ck^2/\log (ak)$ rounds
\gst{regardless of its ID}.
\gst{
Consider first $ak$ transmissions of a station woken up 
randomly and uniformly in $[1,ck^2b^2/\log (bk)]$.
Each of these transmissions is present
in each round $t\in [ck^2a^2/\log(ak),ck^2b^2/\log(bk)]$
with probability $\log(bk)/(ck^2b^2)$. 
Because of this any station woken up
randomly and uniformly in $[1,ck^2b^2/\log (bk)]$,
increases $\sigma(t)$ by at least $a\log(bk)/(ckb^2)$
for any round $t\in [ck^2a^2/\log(ak),ck^2b^2/\log(bk)]$.
Thus for $x$ stations distributed randomly and uniformly in
$[1,ck^2b^2/\log (bk)]$,
in any round $t\in [ck^2a^2/\log(ak),ck^2b^2/\log(bk)]$
the sum of transmission probabilities $\sigma(t)$ is at least
$ax\log(bk)/(ckb^2)$.
}
%
\end{proof}

From this lemma we can derive the next lemma, which then directly implies Theorem~\ref{lb:unknown_k}.
In Lemma~\ref{lem:lb-final}
we express constant $c=\frac{1}{4+\delta}$ in a different form, more suitable for the proof, mainly,
$c = (1-2\epsilon)/(4+2\epsilon)$.
These two forms are equivalent for $\delta=\frac{10\epsilon}{1-2\epsilon}$, and note that $\delta>0$ could be made arbitrarily small for a sufficiently small $0<\epsilon<1/2$, as assumed in the beginning of the proof.

\begin{lemma}
	\label{lem:lb-final}
	For any $0<\epsilon<1/2$, there exists $k_0\in{\mathbb N}$ such that
	for any $k > k_0$ there is a wakeup pattern of $k$ stations such that,
	there is no successful transmission in rounds
	$i\in\left[1, ck^2/\log k\right]$ 
	for a constant $c = (1-2\epsilon)/(4+2\epsilon)$.
\end{lemma}
\begin{proof}
	Since algorithm $\mathscr{A}$ has worst-case latency \gdm{$ck^2/\log k$,} 
	there exists 
	a constant
	$\kappa_0\in\mathbb N$
	such that for any $k > \kappa_0$ we have 
	\gdm{$t_\mathscr{A}(k) \le ck^2/\log k$}.
	
	{For any $\epsilon > 0$}, let us fix \gdm{any set $X$ of} $k$ stations such that $k\epsilon(1-\epsilon)^2>c\kappa_0^2/\log\kappa_0$.
	Our aim is to define a worst-case wake-up pattern for the $k$ stations
	in $X$.
	
	\gdm{For each station $x\in X$, let $\delta_x$ be the delay between 
		the wake-up time and the first transmission of $x$.}
	\gdm{We start by activating first any station $v$ having the maximum 
		delay between its 
		wake-up time and its first transmission, \textit{i.e.} $v$ is chosen
		such that $\delta_v \ge \delta_x$, for every $x\in X\setminus\{v\}$.
		
	Let us consider the interval of rounds 
	$\left[1 +\delta_v,\ ck^2/\log k\right]$ 
	starting with the first transmission of $v$. 
	In the sequel, our aim will be to find a 
	wake-up time $\tau_x \ge 1$ for each station $x\in X\setminus \{v\}$ such that 
	no station transmits before round $1+\delta_v$ and no successful transmission 
	can occur within $\left[1+\delta,\ ck^2/\log k\right]$.
    Let $x_1\ldots, x_{k-1}$ be any ordering of the stations in $X\setminus \{v\}$.
    We can force a collision (with station $v$) at round $1+\delta_v$ 
    by choosing the wake-up time $\tau_{x_1}$ of $x_1$ so that its first transmission occurs at this time $1 + \delta_v$,
    that is letting $\tau_{x_1} = \delta_v - \delta_{x_1} +1$.
    For stations $x_2$ and $x_3$ we
    let $\tau_{x_2} = \delta_v - \delta_{x_2} + 2$
    and $\tau_{x_3} = \delta_v - \delta_{x_3} + 2$,
    so  that their first transmissions arise at time $2+\delta_v$.	
	Analogously, we can choose the wake-up times of the next pair of stations,
	$x_4$ and $x_5$, in such a way that their first transmissions arise at time 
	$3 + \delta_v$.
}   
	Proceeding in this way, we can arrange the wake-up times of 
	{$2k\epsilon$} stations so that
	\gdm{no station transmits before round $1+\delta_v$ and 
	in each round of interval 
	$[1 + \delta_v, k\epsilon]$
	we have the first transmission of two stations. }
	This assures that a collision occurs in each round of this interval and
	therefore no successful transmission can take place in any round of interval
	$[1,k\epsilon]$.
	
    Our next target will be to distribute randomly the activation times of the
    remaining {$k'=k(1-2\epsilon)$} stations so to have small chances of a successful transmission
    throughout the leftover interval $\left({k\epsilon},\ ck^2/\log k\right]$.   
   
    Let $j'$ be the largest integer such that $ck^2(1-\epsilon)^{2j'}/\log (k(1-\epsilon)^{j'}) \ge k\epsilon$.
    Note that $ck^2(1-\epsilon)^{2j'}\ge ck^2(1-\epsilon)^{2j'}/\log (k(1-\epsilon)^{j'}) \ge k\epsilon$, which  
implies that $k(1-\epsilon)^{j'}$ tends to infinity, when $k$ goes to infinity.
    Consider any partition of the remaining {$k'$} stations 
    in $j'+2$ disjoint subsets $S_0,S_1,\ldots,S_{j'},S_{j'+1}$ 
    such that $|S_0|=k'/2$ and $|S_j| = k'\epsilon(1-\epsilon)^j\log k/(2\log (k(1-\epsilon)^j))$
    for $j=1,\ldots, j'$ {($S_{j'+1}$ is filled with the remaining stations)}.
    {In order to see that this is consistent with a partition of $k'$ stations, 
    we need to prove that $\sum_{j=1}^{j'} |S_j| \le k'/2$.}
   Indeed, for a sufficiently large $k$ and small $\epsilon>0$, we observe:
    \begin{equation}
\label{eq:bound-S}
\sum_{j=1}^{j'} \frac{k'\epsilon(1-\epsilon)^j\log k}{2\log (k(1-\epsilon)^j)}= 
      \sum_{j=1}^{j'} \frac{k'\epsilon(1-\epsilon)^j}{2}+
      \sum_{j=1}^{j'} k'\epsilon(1-\epsilon)^j\frac{\log k-\log (k(1-\epsilon)^j)}{2\log (k(1-\epsilon)^j)}
\ .
    \end{equation}
  The first part of Equation~(\ref{eq:bound-S}) could be upper bounded as follows:
    \[ 
      \sum_{j=1}^{j'} \frac{k'\epsilon(1-\epsilon)^j}{2}\le \frac{k'(1-\epsilon)}{2}
\ .
    \]
   Regarding the second part of Equation~(\ref{eq:bound-S}), when we fix $\epsilon>0$, then for $k$ large enough we have the inequalities stated below.
    To prove them, we first replace power $j$ in the denominator by its upper bound $j'$, then we
use the fact that $\log(k(1-\epsilon)^{j'})$ tends to infinity for $k$ going to infinity, and thus $\frac{-\log(1-\epsilon)}{\log(k(1-\epsilon)^{j'})}$ could be upper bounded by $\epsilon^2$ for a sufficiently large $k$, and finally we use straightforward upper bound and arithmetic:
    \[\sum_{j=1}^{j'} k'\epsilon(1-\epsilon)^j\frac{\log k-\log (k(1-\epsilon)^j)}{2\log (k(1-\epsilon)^j)}\le
      \sum_{j=1}^{j'} k'\epsilon(1-\epsilon)^j\frac{-j\log (1-\epsilon)}{2\log (k(1-\epsilon)^{j'})}\le
      \sum_{j=1}^{j'} \frac{k'j\epsilon^3(1-\epsilon)^j}{2}
    \]
    \[
\le
\sum_{j=1}^{\infty} \frac{k'j\epsilon^3(1-\epsilon)^j}{2}= 
       \frac{k'\epsilon(1-\epsilon)}{2}\le\frac{k'\epsilon}{2}
\ .
    \]
  When putting the two above estimates together to Equation~(\ref{eq:bound-S}), we get the sought upper bound:
    \[
\sum_{j=1}^{j'} |S_j| \le
\sum_{j=1}^{j'} \frac{k'\epsilon(1-\epsilon)^j\log k}{2\log (k(1-\epsilon)^j)}\le
      \frac{k'(1-\epsilon)}{2}+\frac{k'\epsilon}{2}\le \frac{k'}{2}
\ .
    \]

We now describe how to distribute randomly the activation times of the
stations in $S_0 \cup S_1 \cup \ldots \cup S_{j'}$ 
(the remaining stations {in $S_{j'+1}$} are discarded).

The wake-up times of the stations from $S_j$ are distributed randomly and uniformly in the interval
$[1,(1-\epsilon)^{2j}ck^2/\log((1-\epsilon)^jk)]$, for $j=0,1,2,\ldots,j'$.
{By Lemma \ref{interval} applied on parameters 
$a = (1-\epsilon)^j$, $b = (1-\epsilon)^{j-1}$ and $x = |S_0| = k'/2$,
it follows that}
for any $t\in [\frac{(1-\epsilon)^{2j}ck^2}{\log((1-\epsilon)^jk)},\frac{(1-\epsilon)^{2(j-1)}ck^2}{\log((1-\epsilon)^{j-1}k)}]$
the sum of transmission probabilities for the stations in
set $S_0$ is at least 

\[
  \frac{(1-\epsilon)^j (k'/2)\log k}{ck}=
      \frac{k'(1-\epsilon)^j \log k}{2ck} \ .
\]

By the same lemma, in this interval the sum of transmission probabilities for
the stations in $S_{q}$, for $0<q<j$, is at least
\[(1-\epsilon)^j \cdot
  \frac{k'\epsilon(1-\epsilon)^{q}\log k}
    {2\log((1-\epsilon)^{q}k)} \cdot
  \frac{\log((1-\epsilon)^{q}k)}
        {ck(1-\epsilon)^{2q}}=
  \frac{k'\epsilon(1-\epsilon)^{j-q}\log k}{2ck} \ .
\]      
Sets $S_{q}$ for $q\ge j$ do not contribute to this probability sum.
Summing these formulas up we get that for
any 
$t\in [\frac{(1-\epsilon)^{2j}ck^2}{\log((1-\epsilon)^jk)},\frac{(1-\epsilon)^{2(j-1)}ck^2}{\log((1-\epsilon)^{j-1}k)}]$
the sum of transmission probabilities is at least
\[\frac{k'\log k}{2ck}\left((1-\epsilon)^{j}+
   \sum_{q=1}^{j-1}\epsilon(1-\epsilon)^{j-q}
  \right)=\frac{k'\log k}{2ck} (1-\epsilon)
	\ .
\]
Thus the sum of transmission probabilities for any round
$t\in[{k\epsilon},ck^2/\log k]$ is  
\begin{equation}\label{equation:sigma}
    \sigma(t) \ge \frac{k'\log k}{2ck} (1-\epsilon)
		\ .
\end{equation}
{Let us call \textit{jammed} a round in which a collision occurs. }
If the sum (\ref{equation:sigma}) is at least $(2+\epsilon)\log k$, 
\gia{which for $k$ sufficiently large is greater than or equal to 1,}
by {Lemma \ref{bound:prob}} the probability of 
having a success or no transmission in any round $t$
{(\textit{i.e.}, the probability that $t$ is not jammed)} is at most
\[((2+\epsilon)\log k+1)e^{\gdm{1-(2+\epsilon)\log k}}
\ ,
\]
which for sufficiently large $k$ is smaller than $k^{-2}$.
Hence, by the union bound over all rounds in 
$[k\epsilon,ck^2/\log k]$, 
{which for sufficiently large $k$ are less than $k^2$}, 
the probability {that there exists a} not jammed round 
in this interval is smaller {than $1$}.
{This implies that a wake-up pattern 
 in which all these rounds are jammed must exist.}

{It remains to show the value of the constant $c$ such that
$\sigma(t) \ge (2+\epsilon)\log k$. }
We have that $\frac{k'\log k}{2ck}(1-\epsilon)\ge(2+\epsilon)\log k$ for $c\le (1-2\epsilon)/(4+2\epsilon)$.
Note that, by taking a sufficiently small $\epsilon$, we can get $c$ arbitrarily close to $1/4$, i.e., $c=\frac{1}{4+\delta}$ for an arbitrary small $\delta$.
\end{proof}

\section{Algorithms with unknown contention size}\label{s:unknown-k}

In this section, we present upper bounds for the case 
when the size of the contention is
unknown. Our results are almost optimal in view of the
lower bound proved in the previous section.

We start with an algorithm for 
the case when acknowledgments are not available; then, 
in Subsection~\ref{s:unknown-k-ack},
we will show how to improve the performance of this algorithm if the stations
are allowed to switch off after a successful transmission. 

\subsection{Algorithm without acknowledgments}
\label{s:unknown-k-noack}

The following algorithm 
works even if the
stations remain active (do not switch off) after a successful transmission. 
As in Algorithm \ref{alg:suniform}, we use a transmission schedule $\cS$
(composed of $N$ binary arrays $\cS_v$)
which specifies for each active station when to transmit and when to stay silent.

The following protocol is 
executed by any station, starting from 
the time at which it is activated. It
takes as input a constant parameter $b$, the number $N$ and a transmission schedule $\cS$. 
We will show that there exists a constant $b$ such that
for any $N$ there is a schedule $\cS$ for which the latency for $k$ 
\gdm{contending} stations is $O(k^2\log N)$.

The protocol's execution mode is organized
in $16N^2$ consecutive phases, scanned by index  
\gdm{$i = 1,\ldots, 16N^2$}, each one lasting 
\gdm{$T = b \ln N$} time rounds, for $j = 1,2,\ldots,T$. 
\gdm{Recall that $N$ is known, so
	 it can be given as an input parameter to the algorithm.}
The $j$th round of phase $i$ will correspond to round
number \gdm{$(i-1)\cdot T+j$} of $v$'s local clock. The entry of $\cS$ 
corresponding to this round will be denoted $\cS_{v,i,j}$. 
Hence, station $v$ will transmit in the $j$th round of phase $i$
(or equivalently in round $i\cdot T+j$ of its local clock) if and only if 
$\cS_{v,i,j} = 1$.
\gdm{Even though the structure is similar as the one in Algorithm 
	\ref{alg:suniform}, the two algorithms will behave very differently.
Indeed, while the main feature of Algorithm \ref{alg:suniform} is a
slow increase of the frequency of transmissions, Algorithm \ref{alg:non-adaptive}
is characterized by a gradual decrease of the transmission rate, 
from one phase to the next, defined by a small polynomial function 
(see Definition \ref{sched}. For such a reason we call this protocol
\unknown\ (\textit{Small Polynomial Rate Decrease}).}

\begin{algorithm}[h]\label{protocol}
	\caption{\unknown$(b,N)$: executed by station $v$}
	\label{alg:non-adaptive}
	\begin{algorithmic}[1]

        \State{\gia{Compute $\cS_v$}}
 
		\For{$i \gets 1,2,3 \ldots, 16 N^2$}
		\For{$j \gets 1,2,3\ldots, T = b \ln N$}
		\State at round $i\cdot T+j$ of local clock transmit 
		if and only $\cS_{v,i,j} = 1$
		\EndFor
		\EndFor
	\end{algorithmic}
\end{algorithm} 

Although the total number of sets in the input schedule is \gdm{$16N^2b \ln N$}, we will prove that
the algorithm terminates within the time at which the first for-loop has reached index 
$i \le 16k^2$.
Namely, we will prove the existence of a schedule $\cS$ for which any 
station sends its message successfully in some round corresponding to a pair
$(i,j)$ such that $i \leq 16k^2$ and $1\le j\le T$. 
This guarantees that the algorithm terminates
within $16k^2 T = 16bk^2 \lN$ rounds.
The rest of the section is substantially devoted to show the existence of such a schedule $\cS$.
To this aim we proceed as follows.

We start by defining a probabilistic construction of the input schedule $\cS$ (Definition~\ref{sched}). Then, we will prove that such a random schedule guarantees a ``good'' probability of successful transmission within $O(k^2 \ln N)$ rounds for any participating station (Lemma~\ref{prob}). Finally, by a straightforward application of the probabilistic method, we argue that such a ``good'' probability guarantees the existence of a schedule $\cS$ such that for any possible instance of any number $k\leq N$ of participating stations, and any possible setting of activation times for these $k$ stations, protocol \unknown$(N;b,\cS)$ solves the contention with latency $O(k^2 \ln N)$ (Theorem~\ref{derand}).
All the results are meant to hold for $k,N$ and $b$ sufficiently large.

In the following, we will refer to the probabilistic protocol determined 
by using, as an input, the random schedule $\cS$ given in the following definition.

\begin{definition}\label{sched}
	Let $\cS$ be such that
	every $\cS_v$ is probabilistically formed as follows:
	for $1\leq i \leq 3$ and $j = 1,2,3,\ldots, b\lN$,
	we let $\Pr (\cS_{v,i,j}=1) = 1/2$;
	for $i = 4,5,6,\ldots, 16N^2$ and $j = 1,2,3,\ldots, b\lN$,
	we let $\Pr (\cS_{v,i,j}=1) = 1/\sqrt{i}$.
\end{definition}

Notice that the probability of transmission depends only on the round number of the local clock, in particular it does not depend on the ID of the transmitting station:
stations that are activated at the same time will transmit always with the same probability; while two stations $v$ and $w$ that are activated at different times may transmit according to different transmission probabilities depending on the gap between the two respective activation times. For example, if at a given round a station $v$ transmits with probability $1/\sqrt{i}$ for some $i$, another station $w$ that has been activated 
after $v$'s activation, will transmit with probability $1/\sqrt{i'}$ for some 
$i' \geq i$.

The following notation will allow us to measure time, for all the stations,
by referring to the local clock of some station $v$.
Given two integers $\alpha,\beta$, with $\alpha < \beta$,
we define a time interval with respect to $v$'s local clock, as follows:
\begin{equation}\label{interval_v}
	[\alpha,\beta]_v = \{r|\ \text{$r$ is a round of $v$'s local clock }, \alpha\le r \le \beta\ \} \ .
\end{equation}
For any station $w \not= v$, let $\Delta_{v,w}$ be the difference between the
activation times of $v$ and $w$. This value can be positive, negative or zero depending on whether $w$ has been activated after, before or at the same time of $v$, respectively.
For any round $r$ of $v$'s local clock, 
we can evaluate the sum of the transmission probabilities at this round, as follows:
\begin{equation}\label{sum_v}
	\sigma(r) = \sum_{w\in V(r)} p(r-\Delta_{v,w}) \ ,
\end{equation}
where $V(r)$ denotes the set of active stations at the $r$th round of $v$'s local clock.

Let $p(r)$ be the probability that an arbitrary station, following the protocol, 
transmits at round $r$ of its local clock. 
We can evaluate the probability that an arbitrary station $v$ transmits successfully
at the $r$th round of its activity. The following result holds 
\gdm{for every algorithm such that each station transmits with probability at most 
	$1/2$ in each slot}, and will
be used in the next subsection as well.

\begin{lemma}\label{success}
	Fix a station $v$ and let $r$ be the $r$th round of $v$'s local clock.
	If each station transmits with probability at most $1/2$, then
	the probability that $v$ transmits successfully at this round $r$ is larger than
	\[
	p(r)\cdot 4^{-\sigma(r)} \ .
	\]
\end{lemma}
\begin{proof}
	In our algorithm every station starts transmitting with probability $1/2$ and the transmission probabilities will decrease with time. 
	So, we are in the hypothesis of the lemma.
	
	Station $v$ transmits successfully at the $r$th round of its clock, if and only if 
	$v$ transmits at this round while all the other stations stay silent. 
	Therefore, the probability stated in the lemma is 
	\begin{align*}
		p(r)\cdot \prod_{w\not =v}(1-p(r-\Delta_{v,w}))  
		&=  p(r)\cdot \prod_{w\not =v}(1-p(r-\Delta_{v,w}))^{\frac{1}{p(r-\Delta_{v,w}))} \cdot p(r-\Delta_{v,w})} \\ 
		&\geq p(r)\cdot \prod_{w\not =v}(1/4)^{p(r-\Delta_{v,w})} \;\;
		\mbox{ [because $p(r-\Delta_{v,w}) \leq  1/2$]}\\
		&>    p(r)\cdot (1/4)^{\sigma(r)} \ .
	\end{align*}	
	
\end{proof}

\gdm{Our next target is to show in Lemma \ref{sub} that there are sufficiently many 
rounds $r$ such that $\sigma(r) < 1$. To this aim we need first the following upper
bound on $s(r) = p(1) + p(2) + \cdots + p(r)$, the sum of the transmission probabilities 
of an arbitrary station up to the $r$th round of its local clock.
}

\begin{lemma}\label{sum}
	We have $s(r) < 2\sqrt{r b\lN}$.
\end{lemma}
\begin{proof}
	By Definition~\ref{sched}, it follows that 
	in every phase $i\ge 1$ 
	any station executing the protocol 
	transmits with probability 
	$\Pr(\cS_{v,i,j} = 1) = \min \{1/2, 1/\sqrt{i}\}$,
	for all $j = 1,2,\ldots,T$, \textit{i.e.} in all $T=b\lN$ rounds of phase $i$.
	Therefore, we have: 
\begin{eqnarray*} 
	s(r) &\le & T \sum_{i=1}^{\lceil r/T \rceil} \min\{1/2, 1/\sqrt{i}\}\\
	     &<& 2T + T\sum_{i=4}^{\lceil r/T \rceil} \frac{1}{\sqrt{i}}\\
    	 &<& 2T + \dk{T\sum_{i=4}^{\lceil r/T \rceil} \int_{i-2}^{i-1} \frac{1}{\sqrt{x}}\, dx}\\
    	 &<& 2T + \dk{T\int_2^{r/T} \frac{1}{\sqrt{x}}\, dx}\\
	     &=& 2T + \dk{T\left(2\sqrt{\frac{r}{T} } - 2\sqrt{2}\right)}\\
	     &<& 2\sqrt{rT} \ ,
\end{eqnarray*}%
which concludes the proof.	
\end{proof}

 The following lemma states that, in any execution
of the protocol by any station $v$, in at least half of the 
first $\Omega(k^2\ln N)$ rounds, counted
since $v$'s activation,
the sum of transmission probabilities is small,
namely less than 1. 

\begin{lemma}\label{sub}
	Let $v$ be an arbitrary station and fix $i = 16k^2$.
	There are at least \gdm{$(i/2)\cdot b\lN $} rounds \gdm{$r\in [1,i\cdot b\lN]_v$} 
	such that $\sigma(r) < 1$.
\end{lemma}
\begin{proof}
	Fix a station $v$. Since at most $k$ stations can be active in 
	any round \gdm{$r\in [1,i\cdot b\lN ]_v$} we have:
	\[
	\sum_{r=1}^{T\cdot i} \sigma(r) \leq k \cdot s(T\cdot i) \ .
	\]
	By Lemma~\ref{sum},
	\begin{eqnarray*}
		k \cdot s(T\cdot i)
		&<& k \cdot 2\sqrt{T^2 i} 
		\ = \ 
		k\cdot 2 T \sqrt{i} 
		\ = \ k \cdot 2T\sqrt{16k^2} 
		\ = \ 8 k^2 T 
		\ < \ \frac{T\cdot i}{2} \ .
	\end{eqnarray*}
	
	Therefore, 
	in at least \gdm{$(i/2) \cdot b \lN $} rounds 
	\gdm{$r\in [1,i\cdot b\lN]_v$}  the sum $\sigma(r)$ must be smaller than 1.
\end{proof}

Now we can derive a lower bound on the probability of having a
successful transmission for an arbitrary station.

\begin{lemma}\label{prob}
	For a sufficiently large constant $b$,
	Protocol \unknown$(N;b,\cS)$ allows any station $v$ to transmit 
	successfully within the first
	$16bk^2\lN$ rounds of its local clock with probability at least 
	$
	1 - N^{-bk/2}.
	$
\end{lemma}

\begin{proof}
	Let $v$ be an arbitrary station and fix $i = 16k^2$. 
	By Lemma~\ref{sub} there is a set $I$ of $(T\cdot i)/2$ rounds, 
	within the first $T\cdot i$ rounds of $v$'s local clock, 
	such that $\sigma(r) < 1$ for every $r\in I$. 
	
	Recalling Lemma \ref{success}, the probability that
	$v$ transmits successfully in any round $r \in I$ is larger than
	\begin{eqnarray*}
	p(r)\cdot (1/4)^{\sigma(r)} 
		&>& p(r)/4 \;\;\;\;\;\;\;\;\;\;\;\;\mbox{ [because $\sigma(r) < 1$ for every $r\in I$]}\\
		&\geq& p(T\cdot i)/4 \ ,
	\end{eqnarray*}
	where the last inequality follows from the observation that $v$
	transmits with the lowest probability in the last phase.
	
	Therefore, the probability that $v$ does not manage to transmit successfully within the first $T\cdot i$ rounds of its activity is at most
	\begin{eqnarray*}
		\lefteqn{\left( 1 - \frac{p(\gdm{T\cdot i})}{4} \right)^{\frac{T\cdot i}{2}}} \\
		&=& \left(1 - \frac{1}{4 \sqrt{16k^2}} \right)^{\frac{T\cdot i}{2}} 
		= \left(1 - \frac{1}{16 k} \right)^{\frac{b\lN 16k^2}{2}} \\
		&=& \left(1 - \frac{1}{16 k} \right)^{\frac{16k\cdot b\lN k}{2}}\\
		&<& e ^{-b\lN k/2} \\
		&=& N ^{-bk/2} \ ,
	\end{eqnarray*}%
which concludes the proof.
\end{proof}

Finally, 
Theorem~\ref{derand} establishes {the upper bound on the maximum latency} 
of our algorithm and represents the main result of this section.

\remove{
\begin{theorem}\label{derand}
	There exists some constant $b>0$ and a (deterministic) schedule $\cS$ 
	such that algorithm \unknown $(N;b,\cS)$ allows 
	any station $v$ to transmit successfully within 
	$O(k^2\ln N)$ rounds. This implies channel utilization
	at least $\Omega(1/(k\log N))$.
\end{theorem}
}

\begin{proof}[Proof of Theorem~\ref{derand}]
	For every station $v$, let us consider the first 
	\gdm{$T' = 16bk^2\lN$} rounds of its activity.
	During the interval $[1,T']_v$, other stations can be activated. Relatively to its
	activation time, any other station can take a total of $T'+1$ configurations: 
	one corresponds to the case the station is not activated during the interval $[1,T']_v$
	and the other $T'$ correspond to the possible activation times that it can have within that interval.
	
	For every fixed station $v$, out of $N$, there are up to $k-1$ other stations, out of the remaining $N-1$, that can become active within the interval $[1,T']_v$. Therefore, the total number of configurations is less than
	$N\cdot \binom{N(T'+1)}{k}$. 
	
	Recalling Lemma~\ref{prob} and taking the union bound over all possible configurations, we get that the probability that there exists a station $v$ that does not transmit successfully, within
	the first $T' = 16bk^2\lN$ rounds of its activity, is less than 
\[
\hspace*{-22em}
		N\cdot \binom{N(T'+1)}{k} \cdot N ^{-bk/2} \ \ \leq
\]
\[
\hspace*{5em}
\leq \ \ 
		\exp [\ln N + k\ln N +
		k\ln((16bk^2\ln N + 1)e/k) -
		(bk/2) \ln N ] 
		\ \ <  \ \ 1 \ ,
\]
	the last inequality holds for a sufficiently large 
	constant~$b$.
	
	Hence, there exists a constant $b$ and a schedule $\cS$ 
	such that the protocol allows any station to transmit successfully within $T'=O(k^2\ln N)$ rounds from its activation.
\end{proof}

\subsection{Algorithm with acknowledgments}
\label{s:unknown-k-ack}

In this section, we show how to improve the time performance 
of our algorithm \unknown\ taking advantage of the 
acknowledgements upon successful transmissions 
\gdm{and under the additional assumption that 
	$\ln\ln N = O(\ln k)$}.
The decrease of the contention size as stations successfully transmit, 
allows us to use slightly larger transmission probabilities
in the probabilistic construction of $\cal S$ (Definition~\ref{schedack}),
which will in turn translate into shorter schedules.

The following protocol, called 
\unknownack\ (\textit{Small Polynomial Rate Decrease with Acknowledgements}), 
is a slight modification of \unknown.
We will show that there exists an input parameter $c$ such that
for any $N$ and $k$ sufficiently large, 
\gdm{with $\ln\ln N = O(\ln k)$},  
there is a schedule $\cS$ for which our protocol
\unknownack$(N;c,\cS)$ guarantees a
latency of $O(k^2\log N/\log k)$ on any set of $k$ 
contending stations.

\begin{algorithm}[ht]\label{protocolack}
	\caption{\unknownack$(N;c)$: executed by station $v$}
	\label{alg:non-adaptive-ack}
	\begin{algorithmic}[1]

        \State{\gia{Compute $\cS_v$}}
 
		\For{$i = 1,2,3 \ldots, c N^2/\lN$}
		\For{$j=1,2,3\ldots, T=\lceil\ln N\rceil$}
		\State at round $i\cdot T+j$ of local clock transmit 
		if and only  $v\in S_{v,i,j}$
		\EndFor
		\EndFor
	\end{algorithmic}
\end{algorithm}

We proceed similarly as in the analysis of algorithm \unknown.
We will prove that there exists a constant 
$c$ such that
the algorithm terminates within the time at which the first for-loop 
has reached index $i \le ck^2/\ln k$,
which guarantees that it terminates within $O(k^2 \lN/\ln k)$ rounds \dk{as long as $\ln\ln N = O(\ln k)$}.
We start with the probabilistic construction of the input schedule $\cS$.
Then, we prove that such a random schedule guarantees high probability 
of having a successful transmission within $O(k^2 \ln N/\ln k)$ rounds 
(Lemma~\ref{suback}). 
Finally, we show that such a probability guarantees the existence of a schedule such that for any possible wake-up pattern of $k$ contending stations, 
protocol \unknownack$(N;c,\cS)$ solves the contention with latency 
$O(k^2 \ln N/\ln k)$ (see Theorem~\ref{derandack}).
 

\begin{definition}\label{schedack}
	Let $\cS$ be such that
every $\cS_v$ is probabilistically formed as follows:
for $1\leq i \leq 3$ and $j = 1,2,3,\ldots,\lN$,
we let $\Pr (\cS_{v,i,j}=1) = 1/2$;
for $i = 4,5,6,\ldots, cN^2/\lN$ and $j = 1,2,3,\ldots,\lN$,
we let $\Pr (\cS_{v,i,j}=1) = \sqrt{(\ln i)/i}$.	
\end{definition}

In the following, we will refer to the probabilistic protocol determined by using the random schedule
$\cS$ given in Definition~\ref{schedack} as an input for Protocol \unknownack.

As already observed in the analysis of protocol \unknown, the probability of transmission does not depend on the ID of the transmitting station. 
We use also the same notation: $p(r)$ is the probability that an arbitrary station following our probabilistic protocol 
transmits at round $r$ of its local clock and $s(r) = p(1) + p(2) + \cdots + p(r)$ is the sum of transmission probabilities for such an arbitrary station 
up to the $r$th round of its clock.  


The following result, analogous to Lemma~\ref{sum}, establishes an upper bound
on $s(r)$ for any station.

\begin{lemma}\label{smack}
   $s(r) < 2\sqrt{Tr \log r}$.
\end{lemma}
\begin{proof}
		By Definition~\ref{schedack}, it follows that 
	in every phase $i\ge 1$ 
	any station executing the protocol 
	transmits with probability 
	$\Pr(\cS_{v,i,j} = 1) = \min \{1/2, \sqrt{\ln i / i}\}$,
	for all $j = 1,2,\ldots,T$, \textit{i.e.} in all $T=\lN$ rounds of phase $i$.
	Therefore, we have: 
	\begin{eqnarray*}
		s(r) &\le & T \sum_{i=1}^{\lceil r/T \rceil} \min\{1/2, \sqrt{\ln i/i}\}\\
		&<& 2T + T\sum_{i=4}^{\lceil r/T \rceil} \sqrt{\frac{\ln i}{i}}\\
		&<& 2T + T\sqrt{\ln r}\sum_{i=4}^{\lceil r/T \rceil} \sqrt{\frac{1}{i}}\\		
		&<& 2T + T\sqrt{\ln r}\gia{\int_2^{r/T} \sqrt{\frac{1}{x}}\, dx}\\
		&=& 2T + T\sqrt{\ln r}\gia{\left(2\sqrt{\frac{r}{T} } - 2\sqrt{2}\right)}\\
		&<& 2\sqrt{Tr\ln r } \ ,
	\end{eqnarray*}%
	which concludes the proof.
\end{proof}

Recall from the previous subsection the notation that allows us to
measure time with respect to the local clock af any fixed station $v$ 
(see (\ref{interval_v}) and (\ref{sum_v})). We can
prove that in at least half of the the first $\Omega(k^2\lN/\ln k)$ rounds of  
any execution of the protocol by any station $v$, 
the sum of transmission probabilities is $O(\ln k)$.

\begin{lemma}\label{sumack}
	Let $v$ be an arbitrary station and fix $i = ck^2/\ln k$, 
	\gdm{for a sufficiently large input constant $c$}.
    There are at least $\lN \cdot (i/2)$ rounds $r\in [1, \lN\cdot i]_v$ 
	such that $\sigma(r)\leq \gdm{4\sqrt{2}  \ln k}$, \dk{as long as $\ln\ln N = O(\ln k)$}.
\end{lemma}
\begin{proof}
	Fix a station $v$ \gdm{and, recalling the assumption that $\ln\ln N = O(\ln k)$, let $d$ be the smallest constant such that
    $\ln\ln N < d\ln k$.}
	 Since at most $k$ stations can be active in 
	any round $r\in [1,T\cdot i]_v$, we have:
	\[
	\sum_{r=1}^{T\cdot i} \sigma(r) \leq k \cdot s(T\cdot i) \ .
	\]
	By Lemma~\ref{smack},
	\begin{eqnarray*}
	k \cdot s(T\cdot i) 
		&<& k \cdot 2\sqrt{T(T\cdot i)\ln(T\cdot i)} \\
	    &=& 2kT\sqrt{i\cdot \ln(T\cdot i)} \\	
	    &=& 2kT\sqrt{i\cdot \ln(\ln N\cdot ck^2/\ln k)} \\
	    &\le& \gdm{2kT\sqrt{i\cdot(\ln\ln N + \ln(ck^2) )} }\\
	    &<& \gdm{2kT\sqrt{i\cdot(d\ln k + \ln(ck^2) )} }\\	  
	    &<& \gdm{2kT\sqrt{i\cdot(c\ln k + c\ln k )} \text{ for a suffici}}\\	   
	    &=& \gdm{2kT\sqrt{2ic\ln k } }\\	   	
	    &=& \gdm{2kT\sqrt{2 c^2 k^2 } }\\	           
		&=& \gdm{2 \sqrt{2}c k^2 T } \\		
		&=& \gdm{\frac{T\cdot i}{2}\cdot 4\sqrt{2}\ln k} \ .
	\end{eqnarray*}
	
	Therefore, 
	in at least $(T\cdot i)/2 =  \lN \cdot(i/2)$ rounds $r\in [1, \lN\cdot i]_v$ the sum $\sigma(r)$ must be smaller than \gdm{$4\sqrt{2}\ln k$}.
\end{proof}

\dk{For the purpose of simplifying the remainder of this technical analysis, let us fix a temporary variable} $P = N^{-4k}$. The next target is to show that the probability
that a station $v$ transmits successfully within the 
first $\Omega(k^2 \ln N/\ln k)$ rounds of its activity is at least $1-2P$. 
We first need to prove that, with high probability, on a constant fraction of 
these rounds,
the sum $\sigma(r)$ is not larger than 1.

\begin{lemma}\label{suback}
	Let $v$ be an arbitrary station.
	\gdm{For a sufficiently large input constant $c$, we have that, with probability at least $1-P$, a quarter of the rounds $r$ in interval $[1,ck^2\lN/\ln k]_v$
	have $\sigma(r)\leq 1$}, \dk{as long as $\ln\ln N = O(\ln k)$}.
\end{lemma}

\begin{proof}
	Let us fix any station $v$, and
	let $L = ck^2\lN/\ln k$ and \gdm{$\gamma = (4\sqrt{2})$}. 
	By Lemma \ref{sumack}, 
	we can define a set $I = \{u_1,u_2,\ldots,u_{L/2}\}$ 
    of the smallest $L/2$ rounds in the 
	interval $[1,L]_v$ such that $\sigma(r)\leq \gamma\ln k$,
	for all $r\in I$.
	
	We now describe a procedure that, working in two phases,
	produces a random binary sequence $\rho$ of length $L/4$.
	The first phase scans consecutively the rounds 
	$r = u_1,u_2, \ldots, u_{L/2}$
	and produces a corresponding random binary sequence $\rho$.
	If the first phase ends with 
	\gdm{$|\rho| < L/4$, the second phase 
	brings the length of $\rho$ to exactly $L/4$} by
	assigning randomly the remaining bits.	 
	We use an index $r$ for the rounds $u_1,u_2, \ldots, u_{L/2}$ and an index 
	$\ell$ for the bits of $\rho$.

  \vspace{0.5cm}
\begin{minipage}[t]{.9\textwidth}
	{\sc First Phase}.	
		
	Let $\ell = 1$. For $r = u_1,u_2, \ldots, u_{L/2}$ 
	and until $\ell \le L/4$ we do:
	\begin{enumerate}
		\item If $1<\sigma(t)\leq\gamma\ln k$ 
		\begin{enumerate}
			\item[(a)] if there is no successful transmission in round $r$, 
			           we set $\rho_{\ell} = 0$. 
			\item[(b)] if there is a successful transmission in round $r$, 
			we set 
			$$
             \rho_r = 
             \begin{cases}
                1, & \text{with probability}\ k^{-1/2}/q_r \\
                0, & \text{with probability}\ 1-k^{-1/2}/q_r
             \end{cases}
            $$			
			where $q_r$ is the probability of having a successful transmission
			(by any station) in round $r$. 
		\end{enumerate}
	    We increment $\ell = \ell+1$.
	    \item If $\sigma(r)\le 1$ we simply skip round $r$ 
	     (without incrementing $\ell$).
	\end{enumerate}

\vspace{0.1cm}
	{\sc Second Phase}. 

	If $\ell < L/4$ we add the remaining bits as follows:
		1 with probability $k^{-1/2}$ and 0 with probability $1-k^{-1/2}$.
\end{minipage}
\vspace{0.5cm}

\noindent
Note that $q_r$ is defined by the transmission probabilities of the protocol:
\gdm{
\[
 q_r = \sum_{v\in V(r)} p(r) \cdot \prod_{w\neq v} (1-p(r-\Delta_{v,w})).
\] 
}
Therefore,
\gdm{
in each bit of the sequence, the probability of having a 1 assigned in the 
first phase is well defined: $q_r\cdot k^{-1/2}/q_r = k^{-1/2}$. The same probability
holds for the bits assigned in the second phase.}
Hence, in each bit of $\rho$, the probability 
of having a 1 is exactly $k^{-1/2}$, independently of
the values of other bits of the sequence. 

For each bit set to 1 in case 1 of the first phase, 
there is a corresponding round in $\{u_1,u_2, \ldots, u_{L/2}\}$
at which some station successfully transmits.
The number of such successful transmissions cannot
exceed $k$, as we have at most $k$ active stations and each of them can 
transmit successfully only once (after which it switches off).
Therefore, if the number of 1's in $\rho$ exceeds $k$, then the extra bits
must have been added in the second phase, which in turn means that
the first phase ended with $\ell < L/4$. Consequently, there were less
than $L/4$ rounds falling in case 1 of the first phase (as in each of them there is an
increase of $\ell$). This implies that 
in more than $L/4$ rounds $r$ of $[1,ck^2\lN/\ln k]_v$ we must have 
$\sigma(r)\leq 1$ (case 2).

Let $X$ be the random variable denoting the number of 1's in $\rho$.
To conclude the proof, we need to prove that $\Pr(X \le k) \le P$.

The expected number of 1's in $\rho$ is 
\[
\mu = \E X = k^{-1/2}(ck^2\lN/\ln k) = ck^{3/2}\lN/\ln k \ .	
\]

By the Chernoff bound, $\Pr(X\leq (1-\delta)\mu)\leq\exp\left(-\frac{\delta^2\mu}{2}\right)$, for any $0<\delta<1$.
For a sufficiently large $c$, 
\gdm{$k \le ck^{3/2}\ln(N)/(2\ln(k)) = {\mu}/{2}$.}
Therefore, we can write:
\begin{eqnarray*}
	\Pr(X\leq k) &\leq & \Pr(X\leq \mu/2) 
	\ = \ \Pr(X\leq (1-1/2)\cdot \mu)
	\ \leq \ e^{-\mu/8} \ = \  e^{-ck^{3/2}\lN/(8\ln k)} \ .
\end{eqnarray*}
Assuming that $N$ and $k$ are sufficiently large, 
the lemma follows by
observing that for a sufficiently large constant $c$,  
\[
e^{-ck^{3/2}\lN/(8\ln k)} \le N^{-4k} = P \ .
\]
\end{proof}

Now we can derive a lower bound on the probability of successful transmission for
an arbitrary station.

\begin{lemma}\label{probP}

   For a sufficiently large input constant $c$,
   Protocol \unknownack$(N;c,\cS)$ allows any station $v$ to transmit 
	successfully within 
	$ck^2\lN/\ln k$ rounds with probability at least $1 - 2P$, \dk{as long as $\ln\ln N = O(\ln k)$}.
\end{lemma}

\begin{proof}
	Fix any station $v$ and let $r$ be any round of $v$'s local clock,
	\gdm{$r \le ck^2\ln N/\ln k$}, 	
	such that $\sigma(r) \le 1$. 
	Recalling Lemma \ref{success},
	the probability that $v$ transmits successfully in such a round $r$
	is larger than
\begin{eqnarray*}
	 p(r)\cdot 4^{-\sigma(r)} 
	       &\ge& p(r)/4 \;\text{ \;\;\;\;\;\;\;\;\;\;\;\;\;\;\;\;\;\;\;\;\ \ [because we assumed $\sigma(r) \le 1$]} \\
	       &\ge& p(ck^2\lN/\ln k)/4  \;\;\;\;\; \text{ [because $p(r)$ is a \gdm{non-increasing} function]}\\
	       &\ge& \frac{\ln k}{4\sqrt{ck^2}} \ .
\end{eqnarray*}	
	Lemma~\ref{suback} states that, 
\gdm{for a sufficiently large constant $c$, we have that, with probability 
	at least $1-P$, there is a set $J$
of size $|J| = ck^2\lN/(4\ln k)$ rounds, within the first $ck^2\lN/\ln k$ rounds 
of $v$'s local clock, such that $\sigma(r)\leq 1$ 
for every $r\in J$. }

	Therefore, the probability that $v$ does not manage to transmit 
	successfully within the first
	$ck^2\lN/\ln k$ rounds of its activity is at most
	\begin{eqnarray*}
		P+\left( 1 - \frac{\ln k}{4k\sqrt{c}} \right)^{|J|}
	    &=& P+\left( 1 - \frac{\ln k}{4k\sqrt{c}} \right)^
	       {\frac{4k\sqrt{c}}{\ln k} \cdot \frac{ck\lN}{16\sqrt{c}} }\\
		&<& P + e ^{-{\sqrt{c}k\lN}/{16}} \\
		&\le& P + N ^{-\sqrt{c}k/16} \\
		&<& 2P \ ,
	\end{eqnarray*}%
	for $c$ sufficiently large.
	
	Finally, the probability that $v$ transmits successfully in 
	the first $ck^2\lN/\ln k$ rounds of its activity is	
	 at least $1-2P$.
\end{proof}

{We are now ready to prove}
Theorem~\ref{derandack} that establishes 
{the upper bound on the maximum latency} of our algorithm and 
concludes the section.

\remove{
\begin{theorem}\label{derandack}
	There exists some constant $c>0$ and a (deterministic) schedule $\cS$ 
	such that algorithm \unknownack $(N;c,\cS)$ allows  
	any station $v$ to transmit successfully within 
	$O(k^2\lN/\ln k)$ rounds. This implies channel utilization
	at least $\Omega(\ln k/(k\log N))$.
\end{theorem}
}

\begin{proof}[Proof of Theorem~\ref{derandack}]
	For every active station $v$, let us consider the first 
	$L=ck^2\lN/\ln k$ rounds of its activity.
	During the interval $[1,L]_v$, other stations can be activated. Relatively to its
	activation time, any other station can take a total of $L+1$ configurations: 
	one corresponds to the case the station is not activated during the interval $[1,T]_v$
	and the other $T$ correspond to the possible activation times that it can have within that interval.
	
	For every fixed station $v$, out of $N$, there are up to $k-1$ other stations, out of the remaining $N-1$, that can become active within 
	the interval $[1,L]_v$. Therefore, the total number of configurations 
	is less than
	$
	N\cdot \binom{N(L+1)}{k} 
	$.
	
	Recalling Lemma~\ref{probP} and taking the union bound over all possible configurations, we get that the probability that there exists a station $v$ that does not transmit successfully, within
	the first $L$ rounds of its activity, is less than 
	\begin{eqnarray*}
		N\cdot \binom{N(L+1)}{k} \cdot N ^{-4k} 
		\ &\leq& \ \exp \left[\ln N + k\ln N 
		+ k\ln\frac{\frac{ck^2\lN}{\ln k}+1}{k/e} - 4k \ln N\right]
		 \ < \  1 \ ,
	\end{eqnarray*}
	where the last inequality holds for a sufficiently large $k$ and constant $c$. \dk{Note that using Lemma~\ref{probP} requires the assumption $\ln\ln N = O(\ln k)$.}
	
	This implies that there exists a constant $c$ and a schedule $\cS$ 
	such that our protocol \unknown $(N;c,\cS)$ allows any station to transmit successfully within $O(k^2\lN/\ln k)$ rounds from its activation.
\end{proof}

\section{Conclusions and Open Problems}
\label{s:conclusions}

We provided a comprehensive study for 
the deterministic contention resolution problem 
in a shared channel in the general situation in which 
each attached station can become active at any time (asynchronous start).

Our results show which parameters allow efficient contention resolution and which do not.
Surprisingly, they show a substantial impact of the
knowledge of the contention size. 
This is notable as it is known
that for a {\em synchronized channel} this feature {\em does not} asymptotically influence 
the efficiency.

The second implication concerns the impact of acknowledgments:
they exponentially improve deterministic channel utilization if 
(some linear estimate of) $k$ is known, unlike the case of  
randomized algorithms where the improvement is only at most 
polynomial~\cite{DS-17}. The acknowledgements are not particularly helpful in case of unknown contention~size. 

{Non-adaptive algorithms use fixed transmission schedules, which could be naturally translated into codes in the radio channel
-- for every transmission sequence, each occurrence of 1 should be substituted by the ID in order to create a codeword, and the contending codewords are placed on the channel with shifts defined by wake-up patterns; each player gets the feedback vector from the channel and decodes the contending IDs.
In this context our results indicate under which conditions such codes could be efficient.}	

Several questions and directions remain open.
One is to shrink the gap between lower and upper bounds for channel utilization
for both known and unknown contention size.
Another intriguing direction concerns the study of energy aspects, {which are particularly challenging when optimized together with latency and channel utilization. 
\dk{Explicit (i.e., polynomial-time) construction of efficient transmission sequences is a third challenging open problem.
Finally, it is interesting to study a tradeoff between maximum latency and the amount of randomness used by stations. It could also help in making algorithms more constructive.}}

\section*{Acknowledgments}

We thank the anonymous reviewers for several comments that helped to improve this paper.

\bibliographystyle{plain}
\bibliography{bibliography}
\end{document}